\documentclass[10pt,journal,compsoc]{IEEEtran}
\usepackage{amsmath,amsthm,amscd,amssymb,amsbsy,mathtools}
\usepackage{cancel}
\usepackage{enumitem}
\usepackage{array}
\usepackage{color}
\usepackage{booktabs}
\usepackage{multirow}
\usepackage{multicol}
\usepackage{caption}
\usepackage{makecell}

\ifCLASSOPTIONcompsoc
  \usepackage[nocompress]{cite}
\else
  \usepackage{cite}
\fi
\ifCLASSINFOpdf
  \usepackage[pdftex]{graphicx}
  \graphicspath{{./figures/}}
\else
\fi
\ifCLASSOPTIONcompsoc
  \usepackage[caption=false,font=footnotesize,labelfont=sf,textfont=sf]{subfig}
\else
  \usepackage[caption=false,font=footnotesize]{subfig}
\fi
\def\x{\mathbf x}
\def\r{\mathbf r}
\def\y{\mathbf y}
\def\u{\mathbf u}

\def\hx{{\hat{\mathbf x}}}
\def\hr{{\hat{\mathbf r}}}
\def\hy{{\hat{\mathbf y}}}
\def\hz{{\hat{\mathbf z}}}
\def\hp{\hat{p}}

\def\tx{{\tilde{\mathbf x}}}
\def\tmu{\tilde{\mu}}

\def\bbE{\mathbb E}

\def\btheta{{\boldsymbol \theta}}
\def\bphi{{\boldsymbol \phi}}
\def\bvarphi{{\boldsymbol \varphi}}
\def\bmu{{\boldsymbol \mu}}
\def\bsigma{{\boldsymbol \sigma}}

\DeclareMathOperator*{\sgn}{sgn}
\DeclareMathOperator*{\logistic}{logistic}

\newtheorem{proposition}{Proposition}

\newcolumntype{C}[1]{>{\centering}p{#1}}

\begin{document}
% \title{A Unified Framework of Lossless and Near-lossless Image Compression via Deep Lossy Plus Residual Coding}
\title{Deep Lossy Plus Residual Coding for \\Lossless and Near-lossless Image Compression}
% \title{Learning Lossy Plus Residual Coding for Lossless and Near-lossless Image Compression}
% \title{Lossless and Near-lossless Image Compression with Deep Lossy Plus Residual Coding}
% \title{Lossless and Near-lossless Image Compression with Deep Lossy Plus Residual Coding}
% \title{Learning Lossless and Near-lossless Image Compression with Lossy Plus Residual Coding}
% \title{Deep Learning of Lossless and Near-lossless Representations of Images}

\author{Yuanchao~Bai,~\IEEEmembership{Member,~IEEE,}
        Xianming~Liu,~\IEEEmembership{Member,~IEEE,}
        Kai~Wang,
        Xiangyang~Ji,~\IEEEmembership{Member,~IEEE,}
        Xiaolin~Wu,~\IEEEmembership{Fellow,~IEEE,}
        Wen~Gao,~\IEEEmembership{Fellow,~IEEE} % <-this % stops a space
\IEEEcompsocitemizethanks{
\IEEEcompsocthanksitem Yuanchao Bai, Xianming Liu and Kai Wang are with the School of Computer Science and Technology, Harbin Institute of Technology, Harbin, 150001, China, E-mail: \{yuanchao.bai, csxm\}@hit.edu.cn, cswangkai@stu.hit.edu.cn.
\IEEEcompsocthanksitem Xiangyang Ji is with the Department of Automation, Tsinghua University, Beijing, 100084, China, E-mail: xyji@tsinghua.edu.cn
\IEEEcompsocthanksitem Xiaolin Wu is with the Department of Electrical and Computer Engineering, McMaster University, Hamilton, L8G 4K1, Ontario, Canada, Email:xwu@ece.mcmaster.ca.
\IEEEcompsocthanksitem Wen Gao is with the School of Electronics Engineering and Computer Science, Peking University, Beijing, 100871, China, E-mail: wgao@pku.edu.cn. \protect \\
}
\thanks{(Corresponding author: Xianming Liu)}
% <-this % stops an unwanted space
}

\IEEEtitleabstractindextext{%
\begin{abstract}
    Lossless and near-lossless image compression is of paramount importance to professional users in many technical fields, such as medicine, remote sensing, precision engineering and scientific research.  But despite rapidly growing research interests in learning-based image compression, no published method offers both lossless and near-lossless modes. In this paper, we propose a unified and powerful deep lossy plus residual (DLPR) coding framework for both lossless and near-lossless image compression. In the lossless mode, the DLPR coding system first performs lossy compression and then lossless coding of residuals. We solve the joint lossy and residual compression problem in the approach of VAEs, and add autoregressive context modeling of the residuals to enhance lossless compression performance. In the near-lossless mode, we quantize the original residuals to satisfy a given $\ell_\infty$ error bound, and propose a scalable near-lossless compression scheme that works for variable $\ell_\infty$ bounds instead of training multiple networks. To expedite the DLPR coding, we increase the degree of algorithm parallelization by a novel design of coding context, and accelerate the entropy coding with adaptive residual interval. Experimental results demonstrate that the DLPR coding system  achieves both the state-of-the-art lossless and near-lossless image compression performance with competitive coding speed.
\end{abstract}

}

\maketitle

\IEEEdisplaynontitleabstractindextext

% \fi
\IEEEpeerreviewmaketitle

\IEEEraisesectionheading{\section{Introduction}\label{sec:introduction}}
%\IEEEPARstart{T}{his} demo file is intended to serve as a ``starter file''
%for IEEE Computer Society journal papers produced under \LaTeX\ using
%IEEEtran.cls version 1.8b and later.
%% You must have at least 2 lines in the paragraph with the drop letter
%% (should never be an issue)
%I wish you the best of success.

\IEEEPARstart{I}{n} many important technical fields, such as medicine, remote sensing, precision engineering and scientific research, imaging in high spatial, spectral and temporal resolutions is instrumental to discoveries and innovations. As achievable resolutions of modern imaging technologies steadily increase, users are inundated by the resulting astronomical amount of image and video data. For example, pathology imaging scanners can easily produce 1GB or more data per specimen. For the sake of cost-effectiveness and system operability (\emph{e.g.}, real-time access via clouds to high-fidelity visual objects), acquired raw images of high resolutions in multiple dimensions have to be compressed.

Unlike in consumer applications (\emph{e.g.}, smartphones and social media), where users are mostly interested in the appearlingness of decompressed images and can be quite oblivious to compression distortions at the signal level, high fidelity of decompressed images is of paramount importance to professional users in many technical fields. In the latter case, the current gold standard is mathematically lossless image compression.
The Shannon's source coding theorem \cite{shannon1948mathematical} establishes the theoretical foundation of lossless image compression, which proves the lower bound of the expected codelength given real probability distribution of image data, \emph{i.e.}, the information entropy. In practice, the compression performance of any specific lossless image codec depends on how well it can approximate the unknown real probability distribution, in order to approach the theoretical lower bound.
Despite years of research, typical compression ratios of traditional lossless image codecs \cite{weinberger2000loco,calic,sneyers2016flif,alakuijala2019jpeg} are limited between 2:1 and 3:1.
An alternative way to improve the compression performance while keeping the high fidelity of decompressed images is near-lossless image compression \cite{nll1998JEI,ke1998near}. Instead of mathematically lossless, near-lossless image compression imposes strict $\ell_\infty$ constraints on the decompressed images requiring the maximum reconstruction error of each pixel to be no larger than a given tight numerical bound.
By introducing the $\ell_\infty$ constrained error bound, near-lossless image compression can guarantee the reliability of each pixel while break the theoretical compression limit of lossless image compression. When the tight error bound is set to zero, near-lossless image compression is equivalent to lossless image compression. Traditional lossless image codecs, such as JPEG-LS \cite{weinberger2000loco} and CALIC \cite{calic,Wu2000nll}, provide users with both lossless and near-lossless image compression in order to meet the requirements of bandwidth and cost-effectiveness for diverse imaging and vision systems.

With the fast progress of deep neural networks (DNNs), learning-based image compression has achieved tremendous progress over the last five years. However, most of these methods are designed for rate-distortion optimized lossy image compression \cite{hu2022pami}, which cannot realize lossless or near-lossless image compression even with sufficient bit-rates. Recently, a number of research teams embark on developing end-to-end optimized lossless image compression methods \cite{Mentzer2019cvpr,mentzer2020cvpr,kingma2019bitswap,townsend2020hilloc,max2019nips,ho2019compression,zhang2021ivpf,zhang2021iflow}.
These methods take advantage of sophisticated deep generative models, such as autoregressive models \cite{pixelrnn2016icml,pixelcnn,pixelcnn_pp}, flow models \cite{kobyzev2020normalizing} and variational auto-encoder (VAE) models \cite{kingma2013auto}, to learn the unknown probability distribution of given image data and entropy encode the image data to bitstreams according to the learned models. While superior compression performance is achieved beyond traditional lossless image codecs, existing learning-based lossless image methods usually suffer from excessively slow coding speed and can hardly be applied to practical full resolution image compression tasks.
It is also regrettable that, unlike traditional JPEG-LS \cite{weinberger2000loco} and CALIC \cite{calic,Wu2000nll}, no studies (except our recent work \cite{Bai_2021_CVPR}) are carried out on learning-based near-lossless image compression, given its great potential as aforementioned.

In this paper, we propose a unified and powerful deep lossy plus residual (DLPR) coding framework for both lossless and near-lossless image compression, which addresses the challenges of learning-based lossless image compression to a large extent. The remarkable characters of the DLPR coding system includes: \textit{the state-of-the-art lossless and near-lossless image compression performance, scalable near-lossless image compression with variable $\ell_\infty$ bounds in a single network and competitive coding speed on even 2K resolution images.}
Specifically, for lossless image compression, the DLPR coding system first performs lossy compression and then lossless coding of residuals.
Both the lossy image compressor and the residual compressor are designed with advanced neural network architectures. We solve the joint lossy image and residual compression problem in the approach of VAEs, and add autoregressive context modeling of the residuals to enhance lossless compression performance.
Note that our VAE model is different from transform coding based VAE for simply lossy image compression \cite{Balle2017iclr,Balle2018variational} or bits-back coding based VAEs \cite{iclr2019bitback} for lossless image compression.
For near-lossless image compression, we quantize the original residuals to satisfy a given $\ell_\infty$ error bound, and compress the quantized residuals instead of the original residuals.
To achieve scalable near-lossless compression with variable $\ell_\infty$ error bounds, we derive the probability model of the quantized residuals by quantizing the learned probability model of the original residuals for lossless compression, without training multiple networks.
Because residual quantization leads to the context mismatch between training and inference, we propose a scalable quantized residual compressor with bias correction scheme to correct the bias of the derived probability model.
In order to expedite the DLPR coding, the bottleneck is the serialized autoregressive context model in residual and quantized residual compression. We thus propose a novel design of coding context to increase the degree of algorithm parallelization, and further accelerate the entropy coding with adaptive residual interval.
Finally, the lossless or near-lossless compressed image is stored including the bitstreams of the encoded lossy image, the encoded original residuals or the quantized residuals.

In summary, the major contributions of this research are as follows:
\begin{itemize}
    \item We propose a unified DLPR coding framework to realize both lossless and near-lossless image compression. The framework can be interpreted as a VAE model and end-to-end optimized. Though lossless and near-lossless modes have been supported in traditional lossless image codecs, such as JPEG-LS or CALIC, we are the first to support both the two modes in learning-based image compression.
    \item We realize scalable near-lossless image compression with variable $\ell_\infty$ error bounds. Given the $\ell_\infty$ bounds, we quantize the original residuals and derive the probability model of the quantized residuals from the learned probability model of the original residuals for lossless image compression, instead of training multiple networks. A bias correction scheme further improves the compression performance.
    \item  To expedite the DLPR coding system, we propose a novel design of coding context to increase the degree of algorithm parallelization without compromising the compression performance. Meanwhile, we further introduce an adaptive residual interval scheme to reduce the entropy coding time.
\end{itemize}
Experimental results demonstrate that the DLPR coding system achieves both the state-of-the-art lossless and near-lossless image compression performance, and achieves competitive PSNR while much smaller $\ell_\infty$ error compared with lossy image codecs at high bit rates. At the same time, the DLPR coding system is practical in terms of runtime, which can compress and decompress 2K resolution images in several seconds.

Note that this paper is the non-trivial extension of our recent work \cite{Bai_2021_CVPR}. First, this paper focuses on both lossless and near-lossless image compression, rather than simply near-lossless image compression in \cite{Bai_2021_CVPR}. Second, we improve the network architectures of lossy image compressor, residual compressor and scalable quantized residual compressor beyond \cite{Bai_2021_CVPR}, leading to more powerful while concise DLPR coding system. Third, to expedite the DLPR coding system, we introduce a novel design of context coding to increase the degree of algorithm parallelization and an adaptive residual interval scheme to accelerate the entropy coding. Finally, we conduct comprehensive experiments to demonstrate that the resulting DLPR coding system achieves the state-of-the-art lossless and near-lossless image compression performance, significantly outperforms its prototype in \cite{Bai_2021_CVPR}, while enjoys much faster coding speed.

The rest of the paper is organized as follows. We provide a brief review of related works in Sec.\;\ref{sec:related_work}. We theoretically analyze lossless and near-lossless image compression problems, and formulate the DLPR coding framework in Sec.\;\ref{sec:dlprc}. The network architecture and acceleration of DLPR coding framework are presented in Sec.\;\ref{sec:network_architecture}. Experiments and conclusions are in Sec.\;\ref{sec:experiments} and Sec.\;\ref{sec:conclusion}, respectively.

\section{Related Work}
\label{sec:related_work}
This section reviews related works from three aspects, including learning-based lossy image compression, learning-based lossless image compression and near-lossless image compression. Our DLPR coding framework takes advantages of the recent progress of learning-based lossy image compression, and achieves the state-of-the-art performance of lossless and near-lossless image compression.

\subsection{Learning-based Lossy Image Compression}
Early learning-based lossy image compression methods with DNNs are based on recurrent neural networks (RNNs), starting from the work of Toderici \emph{et~al.} \cite{Toderici2016iclr}. Toderici \emph{et~al.} \cite{Toderici2016iclr} proposed a long short-term memory (LSTM) network to progressively encode images or residuals, and achieved multi-rate image compression with the increase of RNN iterations. Following \cite{Toderici2016iclr}, Toderici \emph{et~al.} \cite{toderici2017full} and Johnston \emph{et al.} \cite{johnston2018improved} improved RNN-based lossy image compression by modifying the RNN architectures, introducing LSTM-based entropy coding and adding spatially adaptive bit allocation.

Apart from RNN-based methods, a general end-to-end convolutional neural network (CNN) based compression framework was proposed by Ball\'{e} \emph{et~al.} \cite{Balle2017iclr} and Theis \emph{et~al.} \cite{theis2017iclr}, which can be interpreted as VAEs \cite{kingma2013auto} based on transform coding \cite{goyal2001theoretical}.
In this framework, raw images are transformed to latent space, quantized and entropy encoded to bitstreams at encoder side. At decoder side, the quantized latent variables are recovered from the bitstreams and then inversely transformed to reconstruct lossy images. During training, the compression rates are approximated and minimized with the entropy of the quantized latent variables, while the reconstruction distortion is usually minimized with PSNR or MS-SSIM \cite{wang2003msssim}, leading to the rate-distortion optimization \cite{shannon1948mathematical}.
This framework is followed by most recent learned lossy image compression methods and is improved from three aspects, \emph{i.e.}, transform (network architectures) \cite{cheng2020cvpr,Ma2022pami,zhu2022iclr}, quantization \cite{li2020pami,Mentzer2018cvpr,agustsson2017soft} and entropy coding \cite{Balle2018variational,hu2020coarse,minnen2018nips,qian2021iclr,qian2022iclr,Li2020efficient,minnen2020channel,he2021cvpr}.

Inspired by the recent progress of learned lossy image compression, we propose a DLPR coding framework for both lossless and near-lossless image compression, by integrating lossy image compression with residual compression. The DLPR coding system achieves the state-of-the-art lossless and near-lossless image compression performance with competitive coding speed.

\subsection{Learning-based Lossless Image Compression}
Lossless image compression can usually be solved in two steps: 1) statistical modeling of given image data; 2) encoding the image data to bitstreams according to the statistical model, with entropy encoding tools such as arithmetic coder \cite{arithmetic_coding} or asymmetric numerical systems \cite{duda2009asymmetric}.
Given the strong connections between lossless compression and unsupervised machine learning, deep generative models are introduced to solve the first step of lossless image compression, which is a challenging task due to the complexity of unknown probability distribution of raw images.
There are three dominant kinds of deep generative models used in lossless image compression, including autoregressive models, flow models and VAE models.

\textit{Autoregressive models.} Oord \emph{et~al.} \cite{pixelrnn2016icml,pixelcnn} proposed PixelRNN and PixelCNN that estimated the joint distribution of pixels in an image as the product of the conditional distributions over the pixels. The masked convolution was used to ensure that the conditional probability estimation of the current pixel only depends on the previously observed pixels. Following \cite{pixelrnn2016icml,pixelcnn}, Salimans \emph{et~al.} \cite{pixelcnn_pp} proposed PixelCNN++ that improved the implementation of PixelCNN with several aspects.
Reed \emph{et~al.} \cite{reed2017parallel} proposed multi-scale parallelized PixelCNN that allowed efficient probability density estimation.
Zhang \emph{et~al.} \cite{zhang2021nips} studied the out-of-distribution generalization of autoregressive models and utilized a local PixelCNN for lossless image compression.

\textit{Flow models.} Hoogeboom \emph{et~al.} \cite{max2019nips} proposed a discrete integer flow (IDF) model for lossless image compression that learned rich invertible transforms on discrete image data. The latent variables resulting from IDF were assumed to enjoy simpler distributions leading to efficient entropy coding, and were able to recover raw images losslessly. Following \cite{max2019nips}, Berg \emph{et~al.} \cite{berg2020idfpp} further proposed a IDF++ model improving several aspects of the IDF model, such as the network architecture. In \cite{Ma2022pami}, Ma \emph{et~al.} proposed a wavelet-like transform for lossy and lossless image compression, which can be considered as a special flow model. In \cite{ho2019compression}, Ho \emph{et~al.} proposed a local bit-back coding scheme and realized lossless image compression with continuous flows. In \cite{zhang2021ivpf}, Zhang \emph{et~al.} proposed a invertible volume-preserving flow (iVPF) model to achieve discrete bijection for lossless image compression. Beyond \cite{zhang2021ivpf}, Zhang \emph{et~al.} \cite{zhang2021iflow} further proposed a iFlow model composed of modular scale transforms and uniform base conversion systems, leading to the state-of-the-art performance.

\textit{VAE models.} Townsend \emph{et~al.} \cite{iclr2019bitback} proposed bits-back with asymmetric numeral systems (BB-ANS) that performed lossless image compression with VAE models. The bits-back coding scheme estimated posterior distributions of latent variables conditioned on given images and decoded the latent variables from auxiliary bits accordingly. Kingma \emph{et~al.} \cite{kingma2019bitswap} further generalized BB-ANS with a bit-swap scheme based on hierarchical VAE models, to avoid large amount of auxiliary bits. In \cite{townsend2020hilloc}, Townsend \emph{et~al.} proposed an alternative hierarchical latent lossless compression (HiLLoC) method integrating BB-ANS with hierarchical VAE models, and adopted FLIF \cite{sneyers2016flif} to compress parts of image data as auxiliary bits. Besides bits-back coding scheme, Mentzer \emph{et~al.} \cite{Mentzer2019cvpr} proposed a practical lossless image compression (L3C) model, which can also be translated as a hierarchical VAE model.

Instead of the above mentioned methods, we propose a DLPR coding framework, which can be utilized for lossless image compression and interpreted in terms of VAE models. In \cite{mentzer2020cvpr}, Mentzer \emph{et~al.} used the traditional BPG lossy image codec \cite{bpg} to compress raw images and proposed a CNN model to compress the corresponding residuals, which is a special case of our framework.
We further design our network architecture by integrating the VAE model with an autoregressive context model, leading to superior lossless image compression performance. Meanwhile, we further propose novel design of coding context to increase coding parallelization, making the DLPR coding system practical for real image compression tasks.

\subsection{Near-lossless Image Compression}
Near-lossless image compression requires the maximum reconstruction error of each pixel to be no larger than a given tight numerical bound, \emph{i.e.}, the $\ell_\infty$ error bound. It is a challenging task to realize near-lossless image compression, because the $\ell_\infty$ error bound is non-differentiable and must be strictly satisfied.

Traditional near-lossless image compression methods can be divided into three categories: 1) \textit{pre-quantization}: adjusting raw pixel values to the $\ell_\infty$ error bound, and then compressing the pre-processed images with lossless image compression, such as near-lossless WebP \cite{webp}; 2) \textit{predictive coding}: predicting subsequent pixels based on previously encoded pixels, then quantizing predication residuals to satisfy the $\ell_\infty$ error bound, and finally compressing the quantized residuals, such as, \cite{chen1994near,ke1998near}, near-lossless JPEG-LS \cite{weinberger2000loco} and near-lossless CALIC \cite{Wu2000nll}; 3) \textit{lossy plus residual coding}: similar to 2), but replacing predictive coder with lossy image coder, and both the lossy image and the quantized residuals are encoded, as discussed in \cite{nll1998JEI}.
Compared with learning-based lossy and lossless image compression, learning-based near-lossless image compression is still in its infancy.

In this paper, we propose a DLPR coding framework inspired by traditional lossy plus residual coding, which can be utilized for near-lossless image compression. The DLPR coding framework supports scalable near-lossless image compression with variable $\ell_\infty$ bounds without training multiple networks, and achieves the state-of-the-art compression performance. Recently, Zhang \emph{et~al.} \cite{zhang2020ultra} proposed a learning-based soft-decoding method to improve the reconstruction performance of near-lossless CALIC. Though PSNR is improved, the soft-decoding method cannot strictly guarantee the $\ell_\infty$ error bound.

\section{Deep Lossy Plus Residual Coding}
\label{sec:dlprc}
In this section, we introduce a DLPR coding framework for lossless and near-lossless image compression, by integrating lossy image compression with residual compression.
We theoretically analyze lossless and near-lossless image compression problems, and formulate the DLPR coding framework in terms of VAEs.

\subsection{DLPR coding for Lossless Image Compression}
Lossless image compression guarantees that raw image are perfectly reconstructed from the compressed bitstreams.
Assuming that raw images $\x$'s are sampled from an unknown probability distribution $p(\x)$, the shortest expected codelength of the compressed images with lossless image compression is theoretically lower-bounded by the information entropy \cite{shannon1948mathematical,info_theory}:
\begin{equation}
    H(p) = \bbE_{p(\x)} [-\log p(\x)]
    \label{eq:entropy}
\end{equation}
In practice, the compression performance of any specific lossless image compression method depends on how well it can approximate $p(\x)$ with an underlying model $p_\btheta(\x)$. The corresponding compression performance is given by the cross entropy \cite{shannon1948mathematical,info_theory}:
\begin{equation}
    H(p, p_\btheta) = \bbE_{p(\x)} [-\log p_\btheta(\x)] \ge H(p)
    \label{eq:cross_entropy}
\end{equation}
where \eqref{eq:cross_entropy} holds only if $p_\btheta(\x)=p(\x)$.

In order to approximate $p(\x)$, the latent variable model $p_\btheta(\x)$ is extensively employed for this purpose and is formulated by a marginal distribution:
\begin{equation}
    p_\btheta(\x)=\int p_\btheta(\x,\y)d\y=\int p_\btheta(\x|\y)p_\btheta(\y)d\y
    \label{eq:lvm}
\end{equation}
where $\y$ is an unobserved latent variable and $\btheta$ denote the parameters of this model.
Since directly learning the marginal distribution $p_\btheta(\x)$ with \eqref{eq:lvm} is typically intractable, one alternative way is to optimize the evidence lower bound (ELBO) via VAEs \cite{kingma2013auto}. By introducing an inference model $q_\bphi(\y|\x)$ to approximate the posterior $p_\btheta(\y|\x)$, the logarithm of the marginal likelihood $p_\btheta(\x)$ can be rewritten as:
\begin{equation}
    \log p_\btheta(\x)=
    \underbrace{\bbE_{q_\bphi(\y|\x)}\log\frac{p_\btheta(\x,\y)}{q_\bphi(\y|\x)}}_{ELBO}+\underbrace{\bbE_{q_\bphi(\y|\x)}\log\frac{q_\bphi(\y|\x)}{p_\btheta(\y|\x)}}_{D_{kl}(q_\bphi(\y|\x)||p_\btheta(\y|\x))}
    \label{eq:elbo}
\end{equation}
where $D_{kl}(\cdot||\cdot)$ is the Kullback-Leibler (KL) divergence. $\bphi$ denote the parameters of the inference model $q_\bphi(\y|\x)$. Because $D_{kl}(q_\bphi(\y|\x)||p_\btheta(\y|\x))\ge0$ and $\log p_\btheta(\x)\le0$, ELBO is the lower bound of $\log p_\btheta(\x)$. Thus, we have
\begin{equation}
    \bbE_{p(\x)}[-\log p_\btheta(\x)] \le \bbE_{p(\x)}\bbE_{q_\bphi(\y|\x)}\left[-\log\frac{p_\btheta(\x,\y)}{q_\bphi(\y|\x)}\right]
    \label{eq:neg_elbo}
\end{equation}
and can minimize the expectation of negative ELBO as a proxy for the expected codelength $\bbE_{p(\x)}[-\log p_\btheta(\x)]$.

In order to minimize the expectation of negative ELBO, we propose a DLPR coding framework.
We first adopt lossy image compression based on transform coding \cite{goyal2001theoretical} to compress the raw image $\x$ and obtain its lossy reconstruction $\tx$.
The expectation of negative ELBO can be reformulated as follows:
\begin{equation}
    \bbE_{p(\x)}\bbE_{q_\bphi(\hy|\x)}\left[\cancelto{0}{\log q_\bphi(\hy|\x)}\underbrace{-\log p_\btheta(\x|\hy)}_{Distortion}\underbrace{-\log p_\btheta(\hy)}_{R_\hy}\right]
    \label{eq:lossy_code}
\end{equation}
where $\hy$ is the quantized result of continuous latent representation $\y$ and $\y$ is deterministically transformed from $\x$.
Like \cite{Balle2018variational}, we relax the quantization of $\y$ by adding noise from $\mathcal{U}(-\frac12,\frac12)$, and assume $q_\bphi(\hy|\x)=\prod_{i} \mathcal{U}(y_{i}-\frac12,y_{i}+\frac12)$. Thus, $\log q_\bphi(\hy|\x)=0$ is dropped from \eqref{eq:lossy_code}. For simply lossy image compression, such as \cite{Balle2017iclr,Balle2018variational,minnen2018nips,cheng2020cvpr}, the second term of \eqref{eq:lossy_code} can be regarded as the distortion loss between $\x$ and its lossy reconstruction $\tx$ from $\hy$. The third term can be regarded as the rate loss of lossy image compression. Only $\hy$ needs to be encoded to the bitstreams and stored.

Beyond lossy image compression, we further take residual compression into consideration.
The residual $\r$ is computed by $\r=\x-\tx$. We have the following Proposition\;\ref{prop:prop_1}.
\begin{proposition}
    $p_\btheta(\x|\hy)=p_\btheta(\tx, \r|\hy)=p_\btheta(\r|\tx, \hy)$.
    \label{prop:prop_1}
\end{proposition}
\begin{proof}
For each $\x$ and all $(\tx,\r)$ pairs satisfying $\tx+\r=\x$, we have $p_\btheta(\x|\hy)=\sum_{\tx+\r=\x} p_\btheta(\tx,\r|\hy)$.
Following Bayes' rule, we have $p_\btheta(\tx,\r|\hy)=p_\btheta(\tx|\hy)\cdot p_\btheta(\r|\tx,\hy)$.
Thus, we have $p_\btheta(\x|\hy)=\sum_{\tx+\r=\x} p_\btheta(\tx|\hy)\cdot p_\btheta(\r|\tx,\hy)$.
Because the lossy reconstruction $\tx$ is computed by the deterministic inverse transform of $\hy$, there is only one $\tx$ with $p_\btheta(\tx|\hy)=1$ and the other $\tx$'s are with $p_\btheta(\tx|\hy)=0$. Thus, we can have $p_\btheta(\x|\hy)=p_\btheta(\tx|\hy)\cdot p_\btheta(\r|\tx,\hy)+\sum 0 = p_\btheta(\tx, \r|\hy)= 1\cdot p_\btheta(\r|\tx,\hy)= p_\btheta(\r|\tx,\hy)$.
\end{proof}

Based on \eqref{eq:lossy_code} and Proposition\;\ref{prop:prop_1}, we substitute $p_\btheta(\r|\tx,\hy)$ for $p_\btheta(\x|\hy)$ and achieve the DLPR coding formulation:
\begin{equation}
    \bbE_{p(\x)}\bbE_{q_\bphi(\hy|\x)}\left[\underbrace{-\log p_\btheta(\r|\tx,\hy)}_{R_\r}\underbrace{-\log p_\btheta(\hy)}_{R_\hy}\right]
    \label{eq:lossy_residue_code}
\end{equation}
where the first term $R_\r$ and the second term $R_\hy$ of \eqref{eq:lossy_residue_code} are the expected codelengths of entropy encoding $\r$ and $\hy$ with $p_\btheta(\r|\tx,\hy)$ and $p_\btheta(\hy)$, respectively.
During training, we relax the quantization of $\tx$ by adding noise from $\mathcal{U}(-\frac12,\frac12)$, and have $\log p_\btheta(\tx|\hy)=0$ consistent with the precondition of the Proposition\;\ref{prop:prop_1}.
Because \eqref{eq:lossy_residue_code} is equivalent to the expectation of negative ELBO \eqref{eq:neg_elbo}, the proposed DLPR coding framework is the upper-bound of the expected codelength $\bbE_{p(\x)}[-\log p_\btheta(\x)]$ and can be minimized as a proxy.

Note that no distortion loss of lossy image compression is specified in \eqref{eq:lossy_residue_code}. Therefore, we can embed arbitrary lossy image compressors and minimize \eqref{eq:lossy_residue_code} to achieve lossless image compression. A special case is the previous lossless image compression method \cite{mentzer2020cvpr}, in which the BPG lossy image compressor \cite{bpg} with a learned quantization parameter classifier minimizes $-\log p_\btheta(\hy)$ and a CNN-based residual compressor minimizes $-\log p_\btheta(\r|\tx)$.

\subsection{DLPR coding for Near-lossless Image Compression}
We further extend the DLPR coding framework for near-lossless image compression.
Given a tight $\ell_\infty$ error bound $\tau\in\{1,2,\ldots\}$, near-lossless methods compress a raw image $\x$ satisfying the following distortion constraint:
\begin{equation}
    D_{nll}(\x,\hx)=\|\x-\hx\|_\infty=\max_{i,c} |x_{i,c}-\hat{x}_{i,c}|\le\tau
    \label{eq:l_infty}
\end{equation}
where $\hx$ is the near-lossless reconstruction of the raw image $\x$. $x_{i,c}$ and $\hat{x}_{i,c}$ are the pixels of $\x$ and $\hx$, respectively. $i$ denotes the $i$-th spatial position in a pre-defined scan order, and $c$ denotes the $c$-th channel. If $\tau=0$, near-lossless image compression is equivalent to lossless image compression.

In order to satisfy the $\ell_\infty$ constraint \eqref{eq:l_infty}, we extend the DLPR coding framework by quantizing the residuals.
First, we still obtain a lossy reconstruction $\tx$ of the raw image $\x$ through lossy image compression. Although lossy image compression methods can achieve high PSNR results at relatively low bit rates, it is difficult for these methods to ensure a tight error bound $\tau$ of each pixel in $\tx$.
We then compute the residual $\r=\x-\tx$ and suppose that $\r$ is quantized to $\hr$. Let $\hx=\tx+\hr$, the reconstruction error $\x-\hx$ is equivalent to the quantization error $\r-\hr$ of $\r$. Thus, we adopt a uniform residual quantizer whose bin size is $2\tau+1$ and quantized value is \cite{weinberger2000loco,Wu2000nll}:
\begin{equation}
    \hat{r}_{i,c}=\sgn(r_{i,c})\cdot(2\tau+1)\left\lfloor\frac{|r_{i,c}|+\tau}{2\tau+1}\right\rfloor
    \label{eq:r_quantization}
\end{equation}
where $\sgn(\cdot)$ denotes the sign function. $r_{i,c}$ and $\hat{r}_{i,c}$ are the elements of $\r$ and $\hr$, respectively.
With \eqref{eq:r_quantization}, we now have $|r_{i,c}-\hat{r}_{i,c}|\le\tau$ for each $\hat{r}_{i,c}$ in $\hr$, satisfying the tight error bound \eqref{eq:l_infty}. Because residual quantization \eqref{eq:r_quantization} is deterministic, the DLPR coding framework for near-lossless image compression can be formulated as:
\begin{equation}
    \bbE_{p(\x)}\bbE_{q_\bphi(\hy|\x)}\left[\underbrace{-\log p_\btheta(\hr|\tx,\hy)}_{R_\hr}\underbrace{-\log p_\btheta(\hy)}_{R_\hy}\right]
    \label{eq:lossy_residue_code_nll}
\end{equation}
where the first term $R_\hr$ of \eqref{eq:lossy_residue_code_nll} is the expected codelength of the quantized residual $\hr$ with $p_\btheta(\hr|\tx,\hy)$, rather than that of the original residual $\r$ in \eqref{eq:lossy_residue_code}.
Finally, we concatenate the bitstream of $\hr$ with that of $\hy$, leading to the near-lossless image compression result.

\section{Network Architecture and Acceleration}
\label{sec:network_architecture}
\subsection{Network Architecture of DLPR Coding}
\label{subsec:network_arch}
We propose the network architecture of our DLPR coding framework including a lossy image compressor (LIC), a residual compressor (RC) and a scalable quantized residual compressor (SQRC), as illustrated in Fig.\;\ref{fig:network_architecture}.
With LIC and RC, we realize DLPR coding for lossless image compression. With LIC and SQRC, we further realize DLPR coding for near-lossless image compression with variable $\ell_\infty$ bounds $\tau\in\{1,2,\ldots\}$ in a single network, instead of training multiple networks for different $\tau$'s.
We next specify each of the three components in the following subsections.

\begin{figure*}[!t]
\centering
\includegraphics[width=0.99\linewidth]{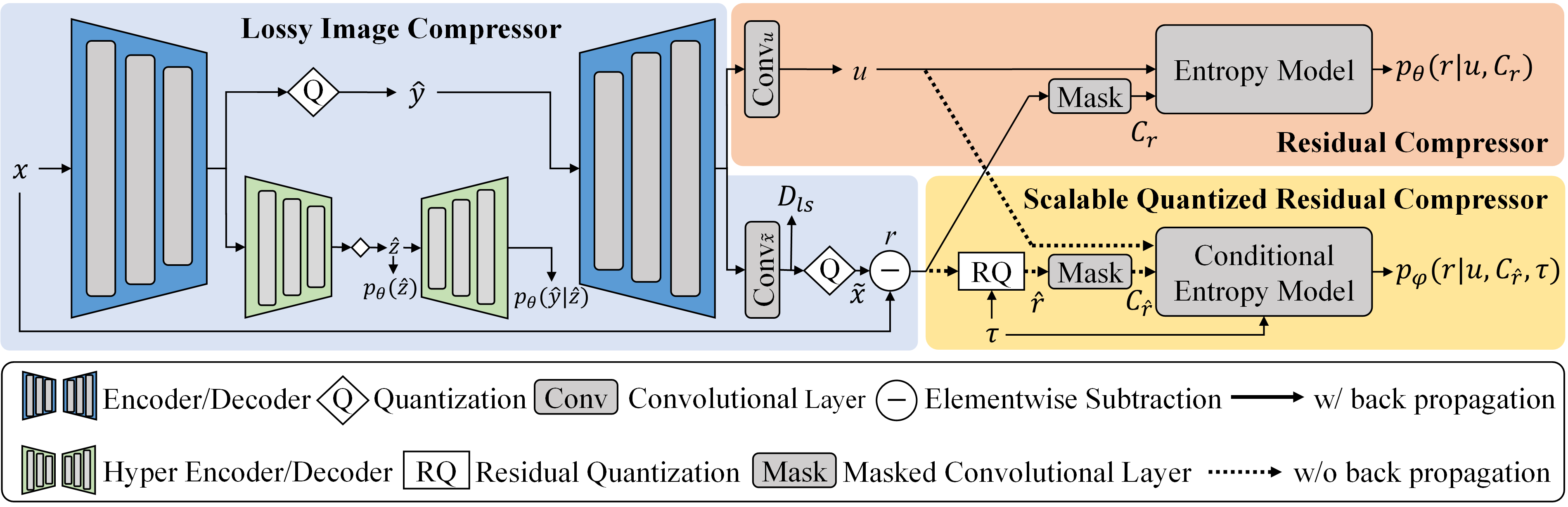}
\caption{Network architecture of DLPR coding framework, including a lossy image compressor (LIC), a residual compressor (RC) and a scalable quantized residual compressor (SQRC).}
\label{fig:network_architecture}
\end{figure*}

\subsubsection{Lossy Image Compressor}
In LIC, we employ sophisticated image encoder and decoder while efficient hyper-prior model \cite{Balle2018variational}, as shown in Fig.\;\ref{fig:lossy_compressor}.
The image encoder $g_e(\cdot)$ and decoder $g_d(\cdot)$ are composed of analysis, synthesis and Swin-Attention blocks following the philosophy of residual and self-attention learning \cite{he2016deep,yulun2019iclr,cheng2020cvpr}, as detailed in Fig.\;\ref{fig:detailed_blocks}.
In Swin-Attention blocks, we adopt the window and shifted window based multi-head self-attention (W-MSA/SW-MSA) in \cite{liu2021swin} to aggregate information within and across local windows adaptively, which improve the representation ability of $g_e(\cdot)$ and $g_d(\cdot)$ with moderate computational complexity.
With $g_e(\cdot)$ and $g_d(\cdot)$, we transform an input raw image $\x$ to its latent representation $\y=g_e(\x)$, quantize $\y$ to $\hy=Q(\y)$, and reversely transform $\hy$ to the lossy reconstruction $\tx=g_d(\hy)$.

Because the sophisticated image encoder and decoder can largely reduce the spatial redundancies in $\x$, the burden of entropy coding $\hy$ is relieved and we decide to employ the efficient hyper-prior model \cite{Balle2018variational} without any context models to ensure the coding parallelization on GPUs. The hyper-prior model extracts side information $\hz=Q(h_e(\y))$ to model the probability distribution of $\hy$, where $h_e(\cdot)$ is the hyper-encoder.
We assume a factorized Gaussian distribution model $\mathcal{N}(\bmu, \bsigma)=h_d(\hz)$ for $p_\theta(\hy|\hz)$ where $h_d(\cdot)$ is the hyper-decoder, and a non-parametric factorized density model for $p_\theta (\hz)$.
The $R_\hy$ in \eqref{eq:lossy_residue_code} and \eqref{eq:lossy_residue_code_nll} is thus extended by
\begin{equation}
    R_{\hy,\hz}=\bbE_{p(\x)}\bbE_{q_\bphi(\hy,\hz|\x)}\left[-\log p_\btheta(\hy|\hz)-\log p_\btheta (\hz)\right]
    \label{eq:lossy_hp}
\end{equation}
where $R_{\hy,\hz}$ is the cost of encoding both $\hy$ and $\hz$.

\begin{figure*}[!t]
\centering
\includegraphics[width=0.9\linewidth]{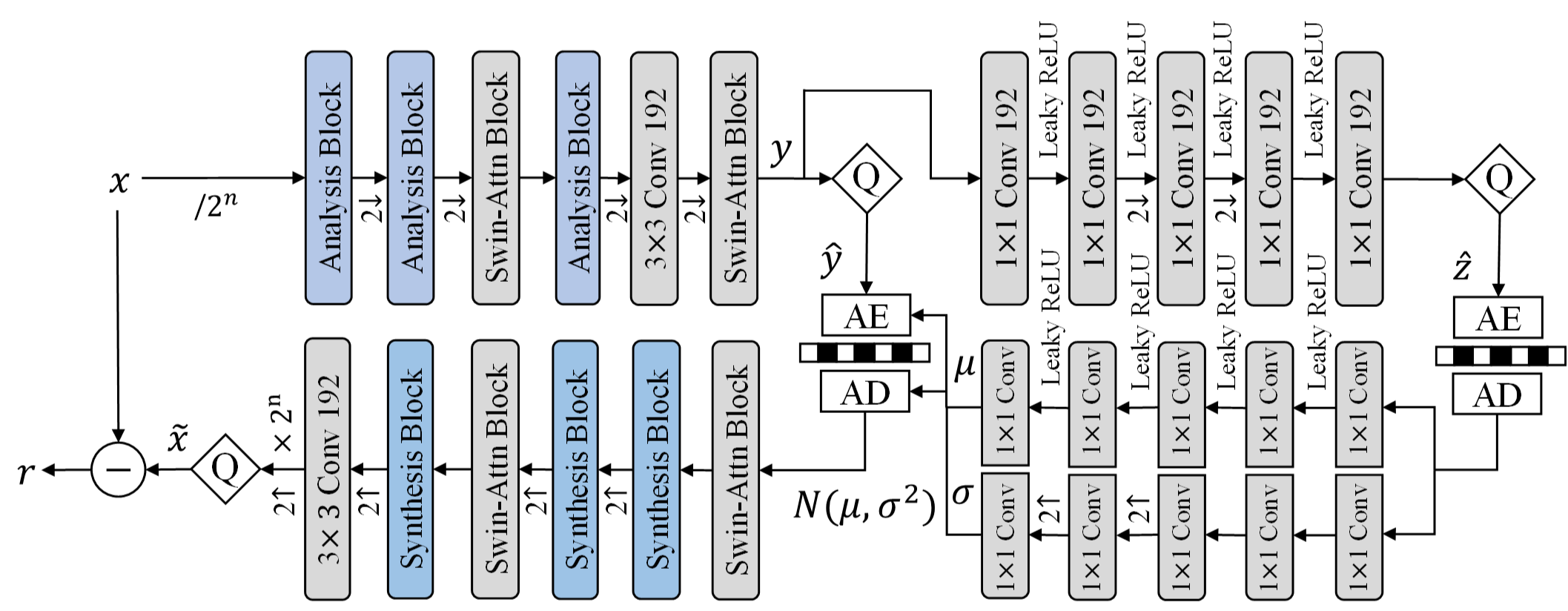}
\caption{Network architecture of lossy image compressor (LIC). We employ sophisticated image encoder/decoder while efficient hyper-prior model. The channel numbers are set uniformly to 192 for all layers. (AE: arithmetic encoding. AD: arithmetic decoding. Q: quantization.)}
\label{fig:lossy_compressor}
\end{figure*}

\begin{figure}[!t]
\centering
\includegraphics[width=0.99\linewidth]{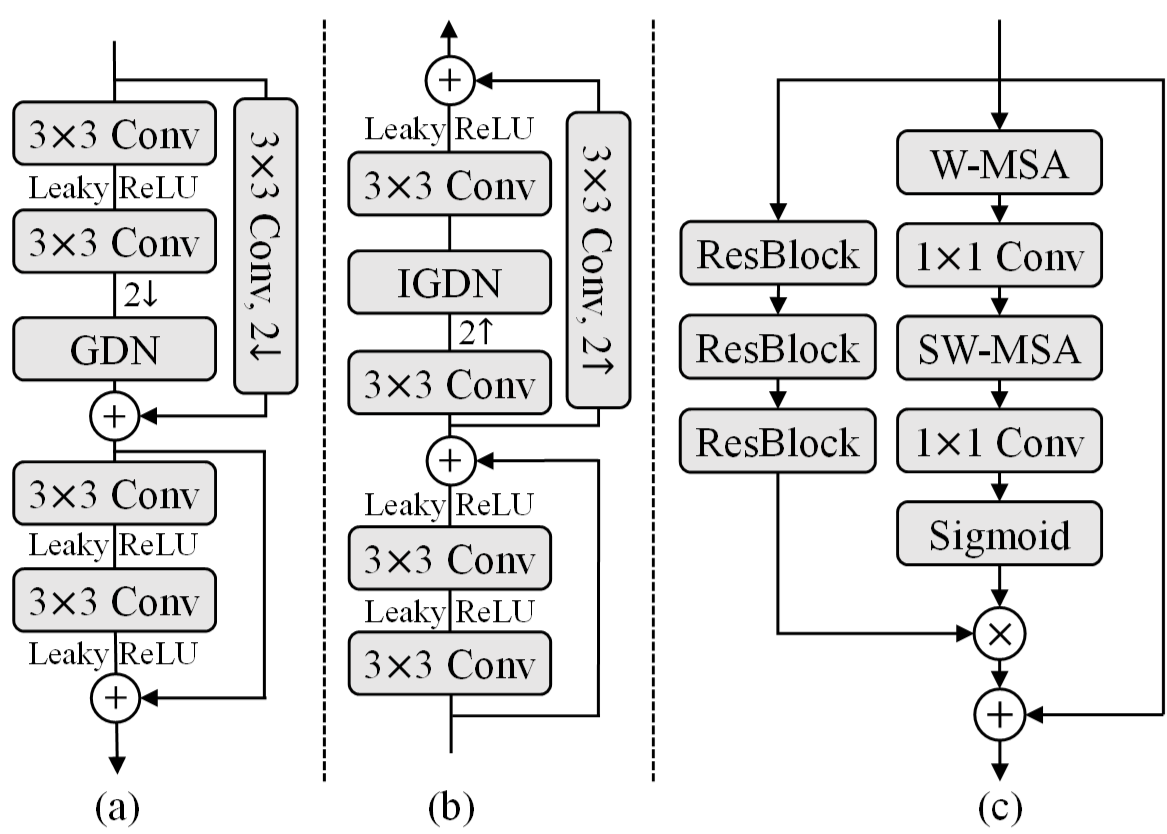}
\caption{Detailed structures of different blocks in LIC. (a) Analysis block. (b) Synthesis block. (c) Swin-Attention block. The window size and head number of Swin-Attention blocks are set to both $8$ for $4\times$ down-sampled feature maps, and are set to 4 and 8 for $16\times$ down-sampled feature maps. The channel numbers are set uniformly to 192 for all layers. (GDN: generalized divisive normalization \cite{gdn2016iclr}. IGDN: inverse GDN. W-MSA/SW-MSA: window/shifted window based multi-head self-attention \cite{liu2021swin}. ResBlock: residual block \cite{he2016deep}.)}
\label{fig:detailed_blocks}
\end{figure}

\subsubsection{Residual Compressor}
\label{subsec:rc}
Given the raw image $\x$ and its lossy reconstruction $\tx$ from LIC, we have the residual $\r=\x-\tx$. We next introduce RC to estimate the probability mass function (PMF) of $\r$ and compress $\r$ with arithmetic coding \cite{arithmetic_coding} accordingly.

Denote by $\u=g_\u(\hy)$, where the feature $\u$ is generated from $\hy$ by $g_\u(\cdot)$. The $g_\u(\cdot)$ and the image decoder $g_d(\cdot)$ share the network except the last convolutional layer, as shown in Fig.\;\ref{fig:network_architecture}. We interpret $\u$ as the feature of the residual $\r$ given $\tx$ and $\hy$. The feature $\u$ shares the same height and width with $\r$ and has 256 channels.
Unlike the latent representation $\hy$ of which the spatial redundancies are largely reduced by the image encoder, the residual $\r$ in the pixel domain has spatial redundancies that cannot be fully exploited by only the feature $\u$. Therefore, we further introduce the autoregressive model into the statistical modeling of $\r$, leading to
\begin{equation}
    p_\btheta(\r|\tx, \hy) = p_\btheta(\r|\u) = \prod_{i,c} p_\btheta (r_{i,c}|\u, r_{<(i,c)})
    \label{eq:p_r_u_ag}
\end{equation}
where $r_{<(i,c)}$ denotes the elements of $\r$ encoded or decoded before $r_{i,c}$ in a pre-defined scan order.
In practice, we implement spatial autoregressive model using a mask convolutional layer with a specific receptive field, rather than depending on all elements in $r_{<(i,c)}$. We regard the receptive field of the mask convolutional layer as the context $C_\r$. Based on \eqref{eq:p_r_u_ag} and $C_\r$, we reformulate the $R_\r$ in \eqref{eq:lossy_residue_code} as
\begin{equation}
    R_{\r}=\bbE_{p(\x)}\bbE_{q_\bphi(\hy,\hz|\x)}\left[-\log p_\btheta (\r|\u, C_\r)\right]
    \label{eq:residual_coding}
\end{equation}

Specifically, we utilize a $7\times 7$ mask convolutional layer with 256 channels to extract the context $C_{r_i}\in C_\r$ from $r_{<(i,c)}$. The $C_{r_i}$ is shared by $r_{i,c}$ of all channels. For RGB images with three channels, we have
\begin{equation}
    p_\btheta(\r|\u,C_\r)=\prod_i p_\btheta(r_{i,1}, r_{i,2}, r_{i,3}|u_i, C_{r_i})
    \label{eq:p_r_channels}
\end{equation}
We further adopt a channel autoregressive scheme over $r_{i,1}$, $r_{i,2}$, $r_{i,3}$ \cite{pixelcnn_pp} and reformulate $p_\btheta(r_{i,1}, r_{i,2}, r_{i,3}|u_i, C_{r_i})$ as
\begin{align}
    p_\btheta(r_{i,1},r_{i,2},r_{i,3}|u_i, C_{r_i})&= p_\btheta(r_{i,1}|u_i, C_{r_i})\cdot  \label{eq:r_chain_rule} \\
    p_\btheta(r_{i,2}|r_{i,1},u_i, C_{r_i}&)\cdot p_\btheta(r_{i,3}|r_{i,1},r_{i,2},u_i, C_{r_i})  \notag
\end{align}
We model the PMF of $r_{i,c}$ with discrete logistic mixture likelihood \cite{pixelcnn_pp} and propose a sub-network to estimate the corresponding entropy parameters, including mixture weights $\pi_{i}^k$, means $\mu_{i,c}^k$, variances $\sigma_{i,c}^k$ and mixture coefficients $\beta_{i,t}$.
$k$ denotes the index of the $k$-th logistic distribution. $t$ denotes the channel index of $\beta$.
The network architecture of the entropy model is shown in Fig.\;\ref{fig:entropy_model}.
We utilize a mixture of $K=5$ logistic distributions.
The channel autoregressive scheme over $r_{i,1},r_{i,2},r_{i,3}$ is implemented by updating the means using:
\begin{align}
    \tmu_{i,1}^k=\mu_{i,1}^k,\quad \tmu_{i,2}^k=\mu_{i,2}^k+\beta_{i,1}\cdot r_{i,1}, \notag \\
    \tmu_{i,3}^k=\mu_{i,3}^k+\beta_{i,2}\cdot r_{i,1}+\beta_{i,3}\cdot r_{i,2}
    \label{eq:channel_ar}
\end{align}
With $\pi_{i}^k$, $\tmu_{i,c}^k$ and $\sigma_{i,c}^k$, we have
\begin{equation}
    p_\btheta(r_{i,c}|r_{i,<c}, u_i, C_{r_i})\sim\sum_{k=1}^{K} \pi_{i}^k \logistic(\tmu_{i,c}^k, \sigma_{i,c}^k)
    \label{eq:factorized_pr}
\end{equation}
where $\logistic(\cdot)$ denotes the logistic distribution. For discrete $r_{i,c}$, we evaluate $p_\btheta(r_{i,c}|r_{i,<c}, u_i, C_{r_i})$ as \cite{pixelcnn_pp}:
\begin{equation}
    \sum_{k=1}^{K} \pi_{i}^k \left[S\left(\frac{r_{i,c}^{+}-\tmu_{i,c}^k}{\sigma_{i,c}^k}\right)-S\left(\frac{r_{i,c}^{-}-\tmu_{i,c}^k}{\sigma_{i,c}^k}\right)\right]
    \label{eq:r_eval}
\end{equation}
where $S(\cdot)$ denotes the sigmoid function. $r_{i,c}^{+}=r_{i,c}+0.5$ and $r_{i,c}^{-}=r_{i,c}-0.5$.
The probability inference scheme of $\r$ is sketched in Fig.\;\ref{fig:rc_sqrc}a.

\begin{figure}[!t]
\centering
\includegraphics[width=0.9\linewidth]{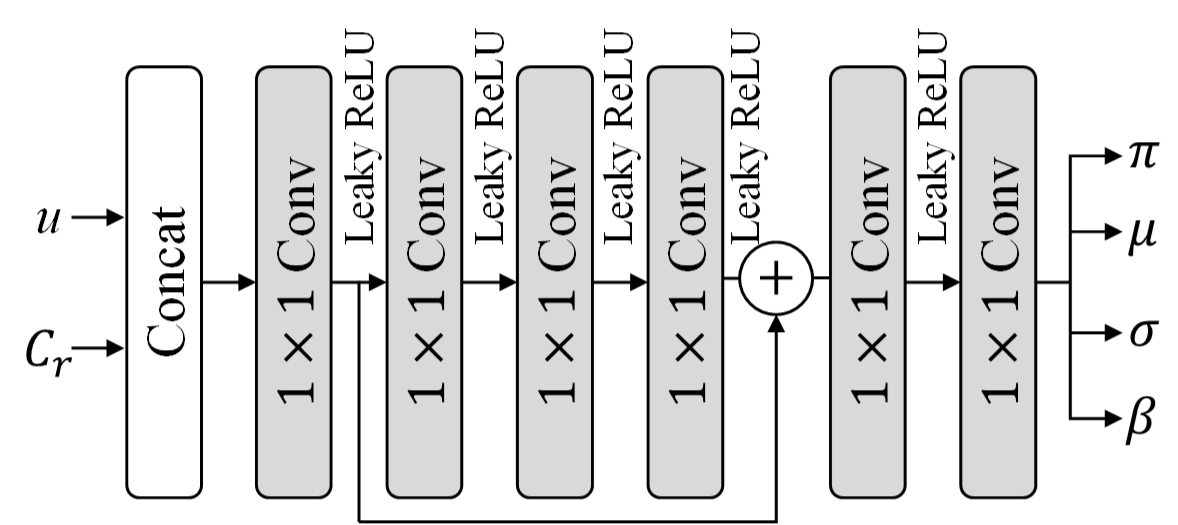}
\caption{Network architecture of entropy model in RC. Given $\u$ and $C_\r$, the entropy model estimates parameters of discrete logistic mixture likelihoods corresponding to the probability distribution of $\r$. All $1\times 1$ convolutional layers except the last layer have 256 channels. The last convolutional layer has $10\cdot K$ channels split by $\pi$, $\mu$, $\sigma$ and $\beta$.}
\label{fig:entropy_model}
\end{figure}

\subsubsection{Scalable Quantized Residual Compressor}
\label{subsubsec:sqrc}
\begin{figure*}[!t]
\centering
\includegraphics[width=0.99\linewidth]{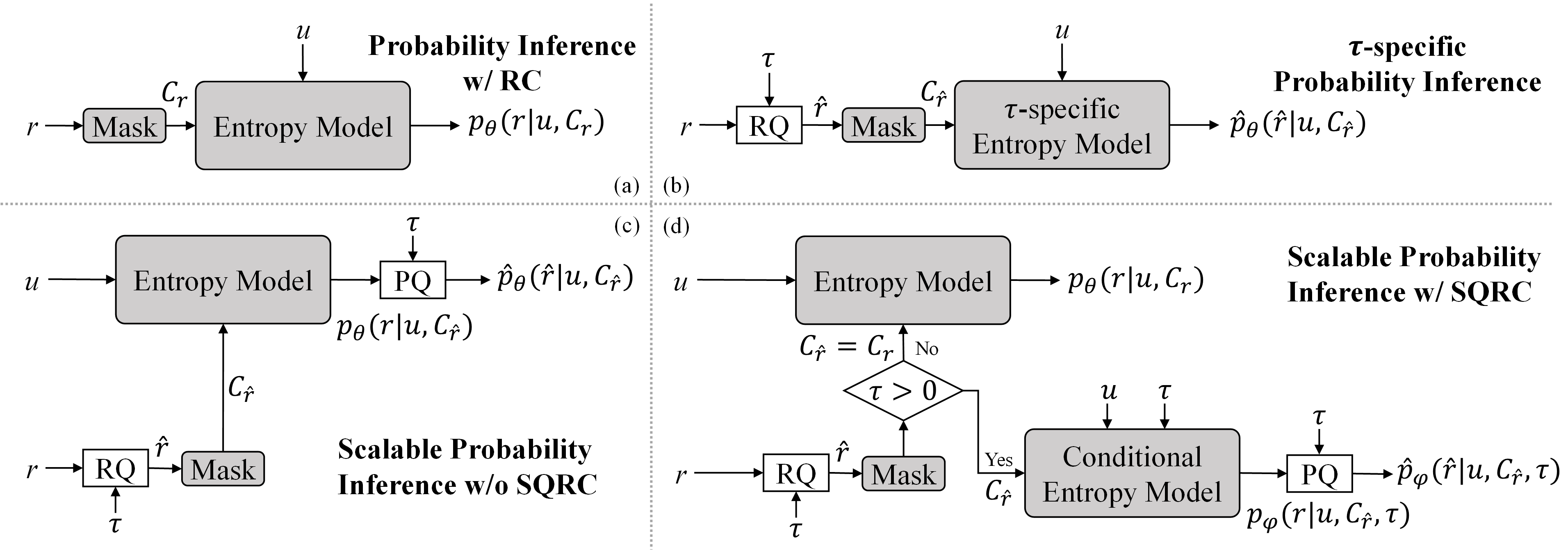}
\caption{Probability inferences of residuals and quantized residuals. (a) Probability inference of residuals with RC. (b) Probability inference of quantized residuals with $\tau$-specific scheme. (c) Scalable probability inference of quantized residuals without SQRC.(d) Scalable probability inference of quantized residuals with SQRC. (RQ: residual quantization. PQ: PMF quantization.)}
\label{fig:rc_sqrc}
\end{figure*}

We finally introduce SQRC to realize scalable near-lossless image compression with variable $\ell_\infty$ bound $\tau\in\{1,2,\ldots\}$.
Though near-lossless image compression given a specific $\tau$ can be realized by optimizing \eqref{eq:lossy_residue_code_nll}, this $\tau$-specific scheme in Fig.\;\ref{fig:rc_sqrc}b leads to two problems:
\begin{itemize}
    \item \textit{Relaxation problem of residual quantization.} Unlike rounding quantization, the bin size of the residual quantization \eqref{eq:r_quantization} is much larger. Moreover, the original residuals are not uniformly distributed in each bin, and thus cannot be relaxed by adding uniform noise.
    \item \textit{Storage problem of multiple networks.} To deploy the near-lossless codec, we have to transmit and store multiple networks for different $\tau$'s, which is storage-inefficient.
\end{itemize}

Instead, we propose a scalable near-lossless image compression scheme, which can circumvent the relaxation of residual quantization and utilize a single network to satisfy variable $\ell_\infty$ error bound $\tau\in\{1,2,\ldots\}$.
Specifically, the scalable compression scheme is based on the learned lossless compression with the DLPR coding framework.
We keep the lossy reconstruction $\tx$ fixed and quantize the original residual $\r$ to $\hr$ with variable $\tau$'s by \eqref{eq:r_quantization}.
To encode the quantized $\hr$, we can derive the PMF of $\hr$ from the learned PMF of the original $\r$. Given $\tau$ and the learned PMF $p_\btheta(r_{i,c}|r_{i,<c}, u_i,C_{r_i})$ of original $r_{i,c}$, the PMF $\hp_\btheta(\hat{r}_{i,c}|r_{i,<c}, u_i,C_{r_i})$ of quantized $\hat{r}_{i,c}$ can be computed by the following PMF quantization:
\begin{equation}
    \hp_\btheta(\hat{r}_{i,c}|r_{i,<c},u_i,C_{r_i})=\sum_{\mathclap{v=\hat{r}_{i,c}-\tau}}^{\hat{r}_{i,c}+\tau}p_\btheta(v|r_{i,<c},u_i,C_{r_i})
    \label{eq:p_rq_ideal}
\end{equation}
We show an illustrative example in Fig.\;\ref{fig:r_quantization}.
Together with \eqref{eq:p_r_channels} and \eqref{eq:r_chain_rule}, we can derive the probability model $\hp_\btheta(\hr|\u,C_\r)$ of $\hr$, which is optimal given the learned $p_\btheta(\r|\u,C_\r)$ of $\r$.
The resulting cost of encoding $\hr$, denoted by $R^\tau_{\hr}$, is reduced significantly with the increase of $\tau$.

However, encoding $\hr$ with $\hp_\btheta(\hr|\u,C_\r)$ results in undecodable bitstreams, since the original residual $\r$ is unknown to the decoder. $\hp_\btheta(\hat{r}_{i,c}|r_{i,<c}, u_i, C_{r_i})$ cannot be evaluated without $r_{i,<c}$ and causal context $C_{r_i}$.
Instead, we can evaluate PMF using the quantized residual $\hr$, \emph{i.e.}, we evaluate $p_\btheta(r_{i,c}|\hat{r}_{i,<c}, u_i, C_{\hat{r}_i})$ and derive $\hp_\btheta(\hat{r}_{i,c}|\hat{r}_{i,<c},u_i,C_{\hat{r}_i})$ with \eqref{eq:p_rq_ideal}, leading to $\hp_\btheta(\hr|\u,C_\hr)$ for the encoding of $\hr$.
Because of the mismatch between training (with $\r$) and inference (with $\hr$) phases, it leads to biased PMF $p_\btheta(r_{i,c}|\hat{r}_{i,<c}, u_i, C_{\hat{r}_i})$, $\hp_\btheta(\hat{r}_{i,c}|\hat{r}_{i,<c},u_i,C_{\hat{r}_i})$ and $\hp_\btheta(\hr|\u,C_\hr)$.
The above probability inference scheme is sketched in Fig.\;\ref{fig:rc_sqrc}c.

\begin{figure}[!t]
\centering
\includegraphics[width=0.95\linewidth]{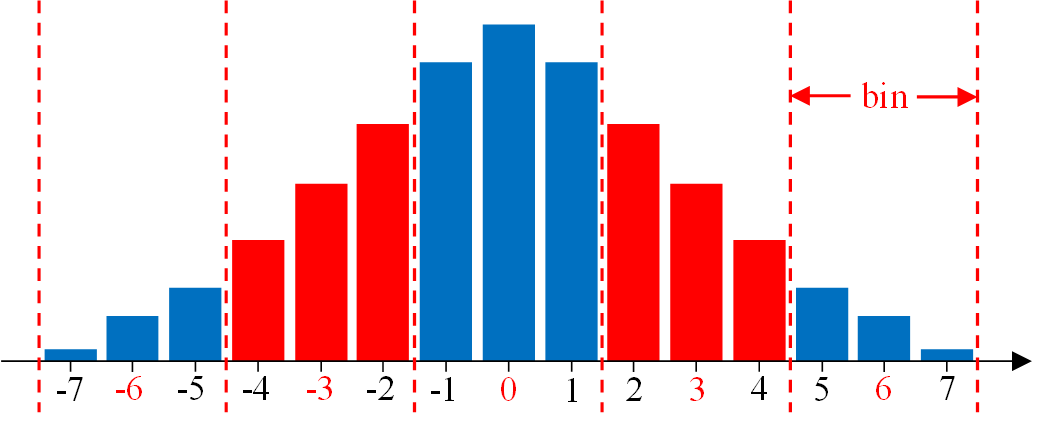}
\caption{PMF quantization corresponding to residual quantization \eqref{eq:r_quantization} with $\tau=1$. Each red number is the value of the quantized residual $\hat{r}_{i,c}$. The probability of each quantized value is the sum of the probabilities of the values in the same bin.}
\label{fig:r_quantization}
\end{figure}

\textbf{SQRC for Bias Correction:} Because of the discrepancy between the oracle $\hp_\btheta(\hr|\u,C_\r)$ and the biased $\hp_\btheta(\hr|\u,C_\hr)$, encoding $\hr$ with $\hp_\btheta(\hr|\u,C_\hr)$ degrades the compression performance.
In order to tackle this problem, we propose SQRC for bias correction to close the gap between the oracle $\hp_\btheta(\hr|\u,C_\r)$ and the biased $\hp_\btheta(\hr|\u,C_\hr)$, while the resulting bitstreams are still decodable.
The components of SQRC are illustrated in Fig.\;\ref{fig:network_architecture}. The masked convolutional layer in SQRC is shared with that in RC. The conditional entropy model has the same network architecture as the entropy model illustrated in Fig.\;\ref{fig:entropy_model}, but replaces the convolutional layers with the conditional convolutional layers \cite{pixelcnn,Choi2019iccv} illustrated in Fig.\;\ref{fig:conditional_cnn}.

The probability inference scheme with SQRC is sketched in Fig.\;\ref{fig:rc_sqrc}d. For $\tau=0$, we still select the entropy model in RC to estimate $p_\btheta(\r|\u, C_{\r})$ to encode $\r$. For $\tau\in\{1,2,\ldots,N\}$, we select the conditional entropy model in SQRC to estimate $p_\bvarphi(\r|\u, C_{\hr}, \tau)$ conditioned on $\tau$. We then derive $\hp_\bvarphi(\hr|\u, C_{\hr},\tau)$ with \eqref{eq:p_rq_ideal} to encode $\hr$, where $\bvarphi$ denote the parameters of SQRC.
As $\hp_\bvarphi(\hr|\u, C_{\hr},\tau)$ approximates the oracle $\hp_\btheta(\hr|\u,C_\r)$ better than the biased $\hp_\btheta(\hr|\u,C_\hr)$, the compression performance can be improved. Since evaluating $\hp_\bvarphi(\hr|\u, C_{\hr},\tau)$ is independent of $\r$, the resulting bitstreams are decodable. In experiments, we demonstrate that the proposed scalable near-lossless compression scheme with SQRC in Fig.\;\ref{fig:rc_sqrc}d can outperform both the $\tau$-specific near-lossless scheme in Fig.\;\ref{fig:rc_sqrc}b and the scalable near-lossless compression scheme without SQRC in Fig.\;\ref{fig:rc_sqrc}c.

\subsection{Training Strategy of DLPR coding}
\begin{figure}[!t]
\centering
\includegraphics[width=0.95\linewidth]{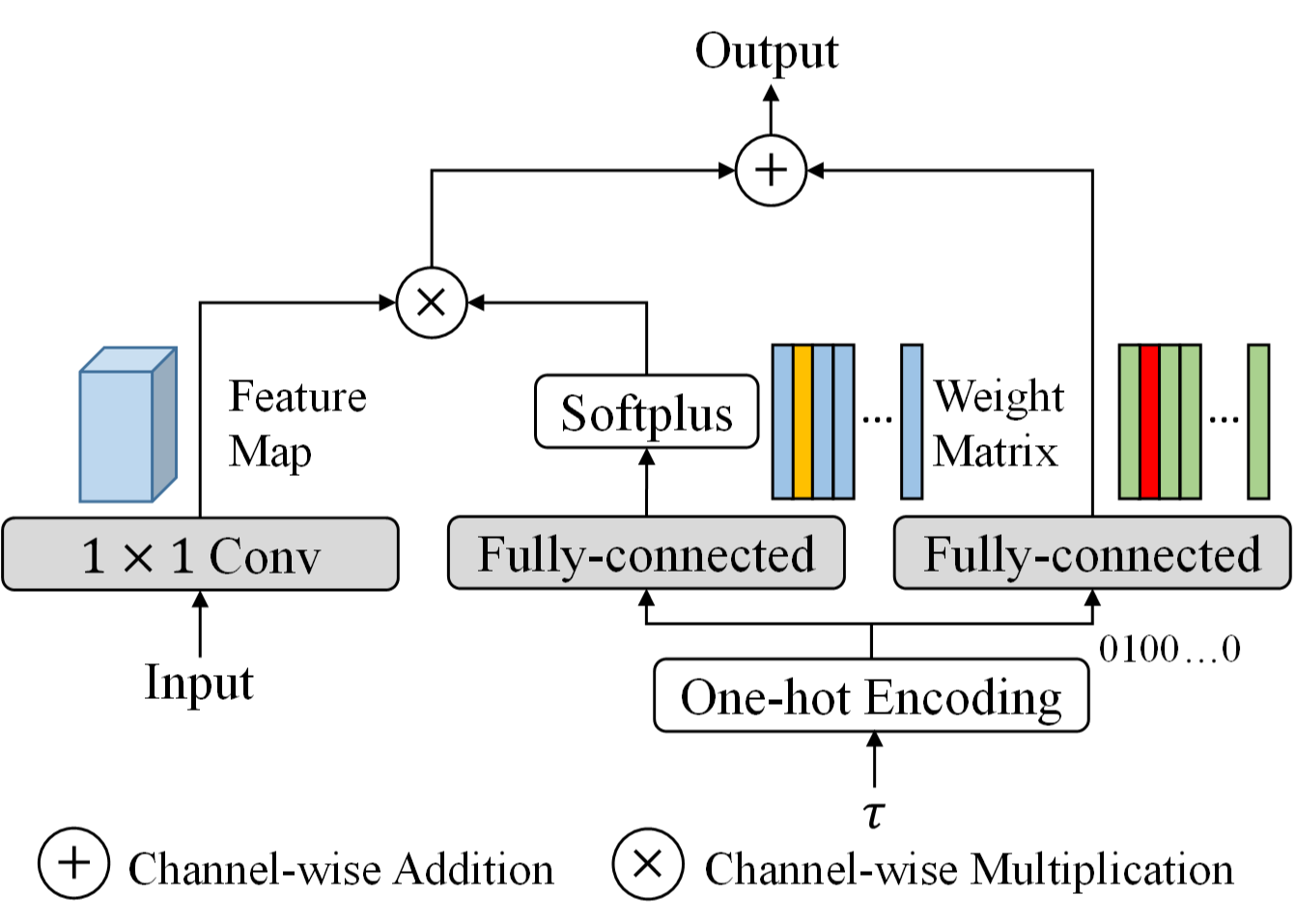}
\caption{Conditional convolutional layer for SQRC. Different outputs can be generated conditioned on $\tau\in\{1,2,\ldots,N\}$. We set $N=5$ in this paper.}
\label{fig:conditional_cnn}
\end{figure}

\subsubsection{Training LIC and RC}
The full loss function for jointly optimizing LIC and RC, \emph{i.e.}, DLPR coding for lossless image compression, is
\begin{equation}
    \mathcal{L}(\btheta, \bphi)=R_{\hy,\hz}+R_{\r}+\lambda\cdot D_{ls}
    \label{eq:loss_func}
\end{equation}
where $\btheta$ and $\bphi$ are the learned parameters of LIC and RC. Besides rate terms $R_{\hy,\hz}$ in \eqref{eq:lossy_hp} and $R_{\r}$ in \eqref{eq:residual_coding}, we further introduce a distortion term $D_{ls}(\x, \tx)$ to minimize the mean square error (MSE) between the raw image $\x$ and its lossy reconstruction $\tx$:
\begin{equation}
    D_{ls}(\x, \tx) = \bbE_{p(\x)}\bbE_{i,c}(x_{i,c}-\tilde{x}_{i,c})^2
    \label{eq:mse_loss}
\end{equation}
As discussed in \cite{Balle2018variational}, minimizing MSE loss is equivalent to learn a LIC that fits residual $\r$ to a factorized Gaussian distribution. However, the discrepancy between the real distribution of $\r$ and the factorized Gaussian distribution is usually large. Therefore, we utilize a sophisticated discrete logistic mixture likelihood model to encode $\r$ in our DLPR coding framework.

The $\lambda$ in \eqref{eq:loss_func} is a ``rate-distortion'' trade-off parameter between the lossless compression rate and the MSE distortion.
When $\lambda=0$, the loss function \eqref{eq:loss_func} is consistent with the theoretical DLPR coding formulation \eqref{eq:lossy_residue_code}, and $\tx$ becomes a latent variable without any constraints.
In experiments, we study the effects of $\lambda$'s on the lossless and near-lossless image compression performance.
We set $\lambda=0$ leading to the best lossless image compression, while set $\lambda=0.03$ leading to robust near-lossless image compression with variable $\tau$'s.

\subsubsection{Training SQRC}
For training SQRC, we generate random $\tau\in\{1,2,\ldots,N\}$ and quantize $\r$ to $\hr$ with \eqref{eq:r_quantization}. Given $\u$ and the extracted context $C_\hr$ from quantized $\hr$, we use the conditional entropy model to estimate $-\log p_\bvarphi(\r|\u, C_{\hr}, \tau)$ conditioned on different $\tau$'s, and minimize
\begin{equation}
    \mathcal{L}(\bvarphi)=\bbE_{p(\x)}\bbE_{q_\bphi(\hy,\hz|\x)}\bbE_{\tau}\left[\log \frac{p_\btheta(\r|\u, C_{\r})}{p_\bvarphi(\r|\u, C_{\hr}, \tau)}\right]
    \label{eq:relative_entropy}
\end{equation}
where $\bvarphi$ denote the learned parameters of the conditional entropy model. $-\log p_\btheta(\r|\u, C_{\r})$ is estimated by the entropy model in RC. $\mathcal{L}(\bvarphi)$ can be considered as an approximate KL-divergence or relative entropy \cite{info_theory} between $p_\btheta(\r|\u, C_{\r})$ and $p_\bvarphi(\r|\u, C_{\hr}, \tau)$.

SQRC is trained together with LIC and RC, but minimizing \eqref{eq:relative_entropy} only updates the parameters of the conditional entropy model as shown in Fig.\;\ref{fig:network_architecture}. The masked convolutional layer is shared with that in RC, and thus can be updated by minimizing \eqref{eq:loss_func}.
This leads to three advantages: 1) We can achieve the target conditional entropy model to close the gap between training with $\r$ and inference with $\hr$; 2) We can circumvent the aforementioned relaxation problem of residual quantization; 3) We can avoid degrading the estimation of $p_\btheta(\r|\u, C_{\r})$ in RC caused by training with randomly generated $\tau$.
Because the entropy model in RC receives the context $C_\r$ extracted from the original residual $\r$, $p_\btheta(\r|\u, C_{\r})$ approximates the true distribution $p(\r|\tx,\hy)$ better than $p_\bvarphi(\r|\u, C_{\hr}, \tau)$. Thus, $-\log p_\btheta(\r|\u, C_{\r})$ is the lower bound of $-\log p_\bvarphi(\r|\u, C_{\hr}, \tau)$ on average.

\subsection{Acceleration of DLPR Coding}
In order to realize practical DLPR coding, the bottleneck is the serialized autoregressive model in RC and SQRC, which severely limits the coding speed on GPUs.
We thus propose a novel design of context coding to increase the degree of algorithm parallelization, and further accelerate the entropy coding with adaptive residual interval.

\subsubsection{Context Design and Parallelization}
\begin{figure}[!t]
\begin{center}
\subfloat[]{
\label{fig:coding_seq}
\includegraphics[width=0.47\linewidth]{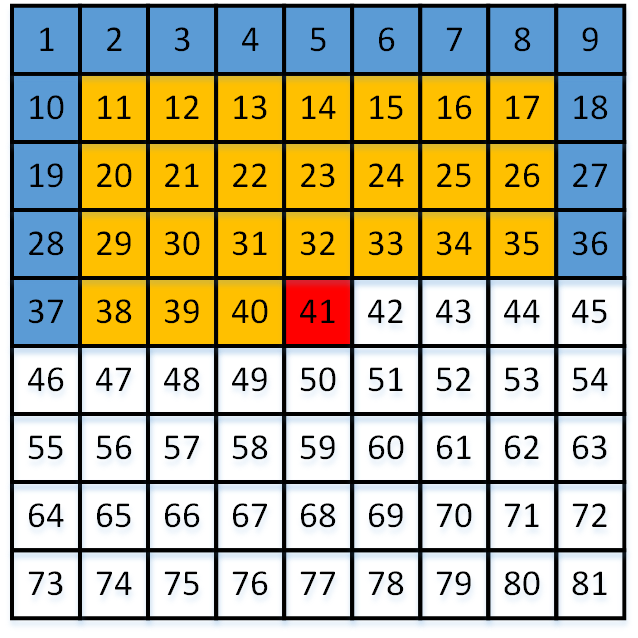}}
\subfloat[]{
\label{fig:coding_acc1}
\includegraphics[width=0.47\linewidth]{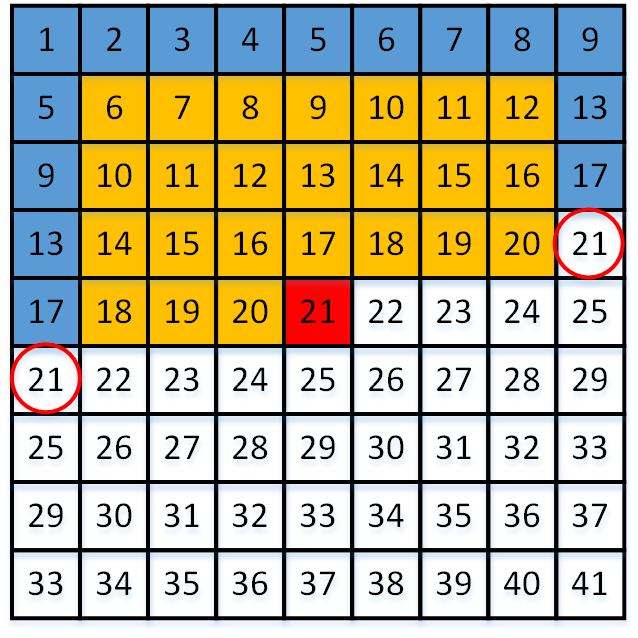}}\\
\subfloat[]{
\label{fig:coding_acc2}
\includegraphics[width=0.47\linewidth]{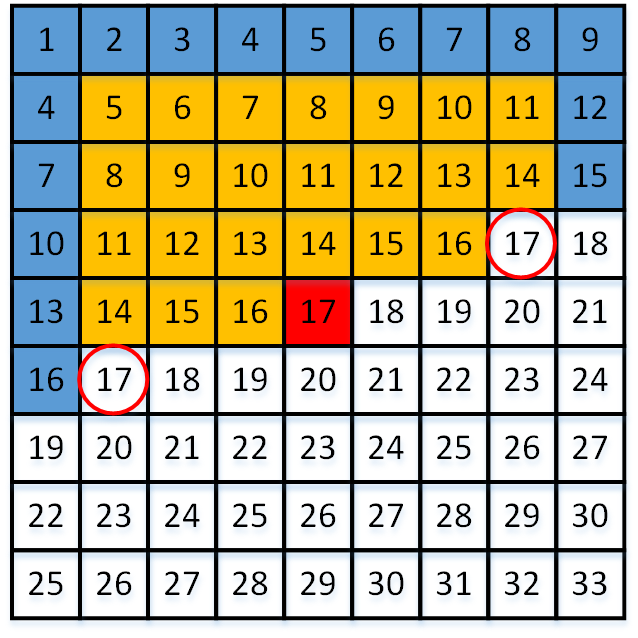}}
\subfloat[]{
\label{fig:coding_acc3}
\includegraphics[width=0.47\linewidth]{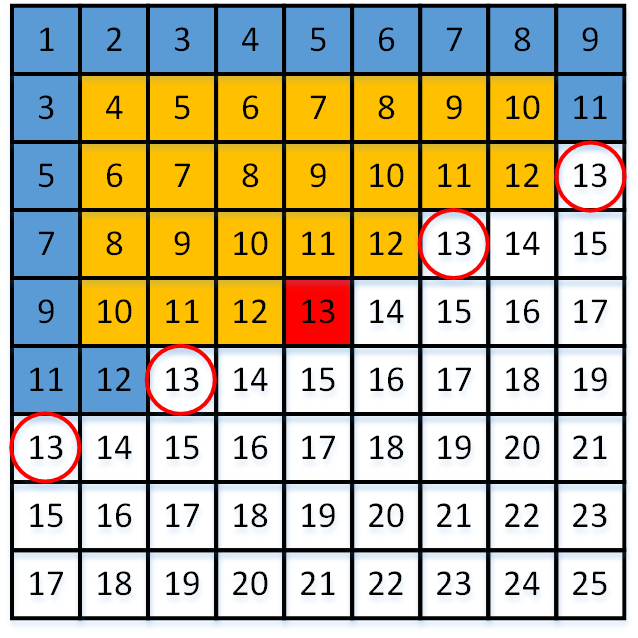}}\\
\subfloat[]{
\label{fig:coding_acc4}
\includegraphics[width=0.47\linewidth]{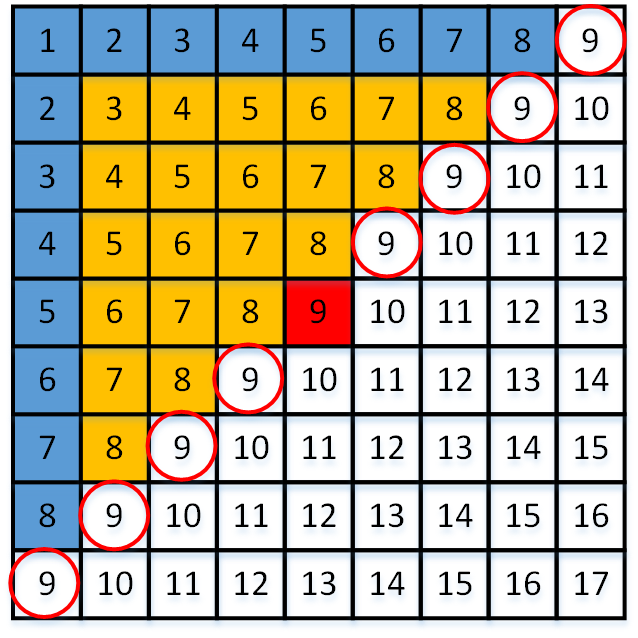}}
\subfloat[]{
\label{fig:coding_acc5}
\includegraphics[width=0.47\linewidth]{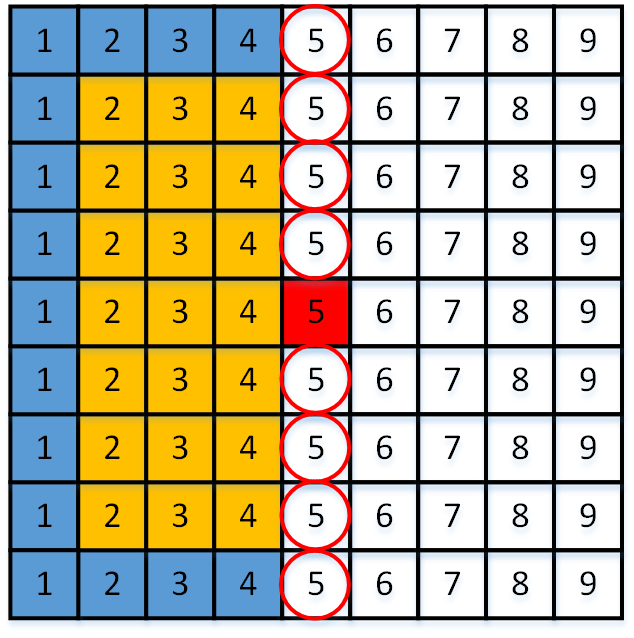}}
\end{center}
\caption{Context design and parallelization (patch size $P=9$, kernel size $k=7$). (a) context model $M_7^5$: raster scan order, $P^2=81$ serial decoding steps. (b) context model $M_7^5$: $14.04^{\circ}$ parallel scan, $5P-4=41$ decoding steps. (c) context model $M_7^4$: $18.43^{\circ}$ parallel scan, $4P-3=33$ decoding steps. (d) context model $M_7^3$: $26.57^{\circ}$ parallel scan, $3P-2=25$ decoding steps. (e) context model $M_7^2$: $45^{\circ}$ parallel scan, $2P-1=17$ decoding steps. (f) context model $M_7^1$: $90^{\circ}$ parallel scan, $P=9$ decoding steps.}
\label{fig:coding_acceleration}
\end{figure}

Generally, autoregressive models suffer from serialized decoding and cannot be efficiently implemented on GPUs. Given an $H\times W$ image, we need to compute $HW$ times context model to decode all pixels sequentially. In lossy image compression, checkerboard context model \cite{he2021cvpr} and channel-wise context model \cite{minnen2020channel} were introduced to accelerate the probability inference of latent variables. However, these two context models are too weak for our residual coding and result in significant performance degradation, without the help of transform coding.

To improve the parallelization of residual coding, we first adopt a common operation to split an $H\times W$ image into multiple non-overlapping $P\times P$ patches and code all $P\times P$ patches in parallel, reducing $HW$ times sequential context computations to $P^2$ times.
We next propose a novel design of context coding to improve the algorithm parallelization given the $P\times P$ patches and $k\times k$ mask convolution, as illustrated in Fig.\;\ref{fig:coding_acceleration}.
Assuming that $P=9$ and $k=7$, we need $P^2=81$ sequential decoding steps in raster scan order for the commonly used context model shown in Fig.\;\ref{fig:coding_seq}.
The number in each pixel denotes the time step $t$ at which the pixel is decoded.
Since $P>\lceil\frac{k}{2}\rceil$ is usually satisfied, the currently decoded pixel only depends on some of the previously decoded pixels.
For example, the red pixel is decoded currently and the yellow pixels are its context. The blue pixels are previously decoded pixels but are not included in the context of the red pixel.
Hence, there are pixels that can be potentially decoded in parallel by revising the scan order.
By using $\frac{180}{\pi}\cdot \arctan(\frac{2}{k+1})$ degree parallel scan, the pixels with the same number $t$ can be decoded simultaneously, as shown in Fig.\;\ref{fig:coding_acc1}. The number of decoding steps is reduced from $P^2$ to $\frac{k+3}{2}\cdot P - \frac{k+1}{2}$. In this case, we use $14.04^{\circ}$ parallel scan, leading to $5P-4=41$ sequential decoding steps. The similar scan order was also used in \cite{zhang2022parallel}.
Moreover, we can remove one context pixel in the upper right of the currently decoded pixel. The newly designed context model leads to $\frac{180}{\pi}\cdot \arctan(\frac{2}{k-1})$ degree parallel scan and reduces decoding steps to $\frac{k+1}{2}\cdot P - \frac{k-1}{2}$.
As shown in Fig.\;\ref{fig:coding_acc2}, we use $18.43^{\circ}$ parallel scan and $4P-3=33$ decoding steps with this context model. When more upper-right context pixels are removed, the coding parallelization can be further improved while the compression performance is gradually compromised on. As shown in Fig.\;\ref{fig:coding_acc3}, the context model leads to $\frac{180}{\pi}\cdot \arctan(\frac{2}{k-3})$ degree parallel scan and $\frac{k-1}{2}\cdot P - \frac{k-3}{2}$ decoding steps, \emph{i.e.}, $26.57^\circ$ parallel scan and $3P-2=25$ decoding steps in our example. When $45^\circ$ parallel scan is reached, this special case is the zig-zag scan \cite{Li2020efficient} and we need $2P-1=17$ decoding steps, as shown in Fig.\;\ref{fig:coding_acc4}. Finally, the fastest case is shown in Fig.\;\ref{fig:coding_acc5}. We can use $90^\circ$ parallel scan and only $P=9$ decoding steps.

In summary, the proposed design of context coding demonstrates that: \textit{given $P\times P$ image patches and $k\times k$ mask convolution with $P>\lceil\frac{k}{2}\rceil$, we can design a series of context models $\{M^{(k+3)/2}_k, M^{(k+1)/2}_k, \ldots, M^1_k\}$ leading to $\{\frac{k+3}{2}\cdot P - \frac{k+1}{2}, \frac{k+1}{2}\cdot P - \frac{k-1}{2}, \ldots, P\}$ parallel decoding steps, by gradually adjusting the context pixels. The corresponding scan angles are $\{\frac{180}{\pi}\cdot \arctan(\frac{2}{k+1}), \frac{180}{\pi}\cdot \arctan(\frac{2}{k-1}), \ldots, 90^\circ\}$, respectively.}
In experiments, we set $P=64$, $k=7$ and select the context model $M^3_7$ shown in Fig.\;\ref{fig:coding_acc3}. The $M^3_7$ enjoys almost the same compression performance as $M_7^5$ in Fig.\;\ref{fig:coding_seq} and Fig.\;\ref{fig:coding_acc1}, but needs much fewer coding steps.

\subsubsection{Adaptive Residual Interval}
Since the pixels of both a raw image $\x$ and its lossy reconstruction $\tx$ are in the interval $[0, 255]$, the element $r_{i,c}$ of the corresponding residual $\r$ is in the interval $[-255, 255]$. To entropy coding each $r_{i,c}$, we need to compute and utilize PMF with $511$ elements, which is relatively large and slows down the entropy coding process.

In practice, the theoretical interval $[-255, 255]$ of $r_{i,c}$ can hardly be filled up. Hence, we can compute and record $r_{\min}=\min_{i,c} r_{i,c}$ and $r_{\max}=\max_{i,c} r_{i,c}$ of each image, and reduce the domain of PMF to the adaptive interval $[r_{\min}, r_{\max}]$ with $r_{\max}-r_{\min}+1$ elements. The overheads of recording the $r_{\min}$ and $r_{\max}$ can be amortized and ignored. For near-lossless image compression with $\tau>0$, we can similarly compute and record the quantized $\hat{r}_{\min}=\min_{i,c} \hat{r}_{i,c}$ and $\hat{r}_{\max}=\max_{i,c} \hat{r}_{i,c}$ of each image, and the domain of the quantized PMF can be further reduced to the adaptive interval $\{\hat{r}_{\min}, \hat{r}_{\min}+2\tau+1, \hat{r}_{\min}+2\cdot(2\tau+1), \ldots, \hat{r}_{\max}\}$, \emph{i.e.}, $\frac{\hat{r}_{\max}-\hat{r}_{\min}}{2\tau+1}+1$ elements in total.
The reduction of residual intervals can significantly accelerate the entropy coding on average.

\section{Experiments}
\label{sec:experiments}
\subsection{Experimental Settings}
We train the DLPR coding system on DIV2K high resolution training dataset \cite{div2k} consisting of 800 2K resolution RGB images. Although DIV2K is originally built for image super-resolution task, it contains large number of high-quality images that is suitable for training our codec.
During training, the 2K images are first cropped into non-overlapped $121379$ patches with the size of $128\times 128$. We then flip these patches horizontally and vertically with a random factor $0.5$, and further randomly crop the flipped patches to the size of $64\times 64$.
We optimize the proposed network for $600$ epochs using Adam \cite{kingma2015adam} with minibatches of size $64$. The learning rate is initially set to $1\times 10^{-4}$ and is decayed by $0.9$ at the epoch $\in[350, 390, 430, 470, 510, 550, 590]$.

We evaluate the trained DLPR coding system on six image datasets:
\begin{itemize}
    \item \textit{ImageNet64}. ImageNet64 validation dataset \cite{chrabaszcz2017downsampled} is a downsampled variant of ImageNet validation dataset \cite{deng2009imagenet}, consisting of $50000$ images of size $64\times 64$.
    \item \textit{DIV2K}. DIV2K high resolution validation dataset \cite{div2k} consists of 100 2K color images sharing the same domain with the DIV2K high resolution training dataset.
    \item \textit{CLIC.p}. CLIC professional validation dataset\footnote{https://www.compression.cc/challenge/\label{ftn:clic}} consists of 41 color images taken by professional photographers. Most images in CLIC.p are in 2K resolution but some of them are of small sizes.
    \item \textit{CLIC.m}. CLIC mobile validation dataset\textsuperscript{\ref{ftn:clic}} consists of 61 2K resolution color images taken with mobile phones. Most images in CLIC.m are in 2K resolution but some of them are of small sizes.
    \item \textit{Kodak}. Kodak dataset \cite{kodak} consists of 24 uncompressed $768\times 512$ color images, widely used in evaluating lossy image compression methods.
    \item \textit{Histo24}. Besides natural images, we build a Histo24 dataset consisting of 24 uncompressed $768\times 512$ histological images, in order to evaluate our codec on images of different modality. These histological images are randomly cropped from high resolution ANHIR dataset \cite{borovec2020anhir}, which is originally used for histological image registration task.
\end{itemize}

The DLPR coding system is implemented with Pytorch. We train the DLPR coding system on NVIDIA V100 GPU, while evaluate the compression performance and running time on Intel CPU i9-10900K, 64G RAM and NVIDIA RTX3090 GPU. We use \textit{torchac} \cite{Mentzer2019cvpr}, an arithmetic coding tool in Pytorch, for entropy coding.

\subsection{Lossless Results of DLPR coding}
\begin{table*}[!t]
\caption{Lossless image compression performance (bpsp) of the proposed DLPR coding system with $\lambda=0$, compared with other lossless image codecs on ImageNet64, DIV2K, CLIC.p, CLIC.m, Kodak and Histo24 datasets.}
\label{tb:results_ll}
\centering
\small
\begin{tabular}{lC{5em}C{5em}C{5em}C{5em}C{5em}C{5em}}
\toprule[1pt]
  Codec     &  ImageNet64 & DIV2K & CLIC.p & CLIC.m & Kodak & Histo24 \tabularnewline
\midrule
    PNG                                 & 5.42  & 4.23 & 3.93 & 3.93 & 4.35 & 3.79    \tabularnewline
    JPEG-LS \cite{weinberger2000loco}   & 4.45  & 2.99 & 2.82 & 2.53 & 3.16 & 3.39    \tabularnewline
    CALIC \cite{calic}                  & 4.71  & 3.07 & 2.87 & 2.59 & 3.18 & 3.48    \tabularnewline
    JPEG2000 \cite{skodras2001j2k}      & 4.74  & 3.12 & 2.93 & 2.71 & 3.19 & 3.36    \tabularnewline
    WebP \cite{webp}                    & 4.36  & 3.11 & 2.90 & 2.73 & 3.18 & 3.29    \tabularnewline
    BPG \cite{bpg}                      & 4.42  & 3.28 & 3.08 & 2.84 & 3.38 & 3.82    \tabularnewline
    FLIF \cite{sneyers2016flif}         & 4.25  & 2.91 & 2.72 & 2.48 & 2.90 & 3.23    \tabularnewline
    JPEG-XL \cite{alakuijala2019jpeg}   & 4.94  & 2.79 & 2.63 & 2.36 & 2.87 & 3.07    \tabularnewline
\midrule
    L3C \cite{Mentzer2019cvpr}          & 4.48  & 3.09 & 2.94 & 2.64 & 3.26 & 3.53     \tabularnewline
    RC  \cite{mentzer2020cvpr}          & $-$   & 3.08 & 2.93 & 2.54 & $-$  & $-$      \tabularnewline
    Bit-Swap \cite{kingma2019bitswap}   & 5.06   & $-$  & $-$  & $-$  & $-$  & $-$     \tabularnewline
    HiLLoC \cite{townsend2020hilloc}    & 3.90  & $-$  & $-$  & $-$  & $-$  & $-$     \tabularnewline
    IDF    \cite{max2019nips}           & 3.90  & $-$  & $-$  & $-$  & $-$  & $-$     \tabularnewline
    IDF++  \cite{berg2020idfpp}         & 3.81  & $-$  & $-$  & $-$  & $-$  & $-$     \tabularnewline
    LBB    \cite{ho2019compression}     & 3.70  & $-$  & $-$  & $-$  & $-$  & $-$     \tabularnewline
    iVPF \cite{zhang2021ivpf}           & 3.75  & 2.68 & 2.54 & 2.39 & $-$  & $-$    \tabularnewline
    iFlow \cite{zhang2021iflow}         & \textbf{3.65}  & 2.57 & 2.44 & 2.26 & $-$  & $-$     \tabularnewline
\midrule
    DLPR (Ours)                 & 3.69  & \textbf{2.55} & \textbf{2.38} & \textbf{2.16} & \textbf{2.86} & \textbf{2.96} \tabularnewline
\bottomrule[1pt]
\end{tabular}

\end{table*}

\begin{table*}[!t]
\caption{Near-lossless image compression performance (bpsp) of the proposed DLPR coding system with $\lambda=0.03$, compared with near-lossless JPEG-LS, near-lossless CALIC and near-lossless WebP on ImageNet64, DIV2K, CLIC.p, CLIC.m, Kodak and Histo24 datasets. $^*$\textit{The error bounds of near-lossless WebP are powers of two.}}
\label{tb:results_nll}
\centering
\small
\begin{tabular}{lcC{5em}C{5em}C{5em}C{5em}C{5em}C{5em}}
\toprule[1pt]
  Codec & $\tau^*$ & ImageNet64 & DIV2K & CLIC.p & CLIC.m & Kodak & Histo24\tabularnewline
\midrule
    \multirow{3}{*}{\makecell{WebP nll\cite{webp}}}    & 1  & 3.61  & 2.45 & 2.26 & 2.11 & 2.41 & 2.31 \tabularnewline
                                                       & 2  & 3.11  & 2.04 & 1.89 & 1.85 & 2.01 & 1.76 \tabularnewline
                                                       & 4  & 2.70  & 1.83 & 1.73 & 1.75 & 1.82 & 1.73 \tabularnewline
\midrule
    \multirow{3}{*}{\makecell{JPEG-LS\cite{weinberger2000loco}}}     & 1  & 4.01 & 2.62 & 2.34 & 2.44 & 2.90 & 1.99 \tabularnewline
                                                                     & 2  & 3.25 & 2.07 & 1.80 & 1.89 & 2.30 & 1.58 \tabularnewline
                                                                     & 4  & 2.49 & 1.53 & 1.28 & 1.35 & 1.68 & 1.24 \tabularnewline
\midrule
    \multirow{3}{*}{\makecell{CALIC\cite{Wu2000nll}}}       & 1  & 3.69  & 2.45 & 2.18 & 2.28 & 2.75 & 1.78 \tabularnewline
                                                            & 2  & 2.94  & 1.88 & 1.62 & 1.70 & 2.14 & 1.28 \tabularnewline
                                                            & 4  & 2.41  & 1.31 & 1.07 & 1.13 & 1.51 & 0.84 \tabularnewline
\midrule
    \multirow{3}{*}{\makecell{DLPR (Ours)}}  & 1 & 2.59  & 1.69 & 1.56 & 1.50 & 1.81 & 1.71 \tabularnewline
                                             & 2 & 2.06  & 1.26 & 1.13 & 1.09 & 1.37 & 1.23 \tabularnewline
                                             & 4 & 1.55  & 0.84 & 0.69 & 0.67 & 0.90 & 0.65 \tabularnewline
\bottomrule[1pt]
\end{tabular}
\end{table*}

We evaluate the lossless image compression performance of the proposed DLPR coding system, measured by \textit{bit per subpixel} (bpsp). Each RGB pixel has three subpixels. We compare with eight traditional lossless image codecs including PNG, JPEG-LS \cite{weinberger2000loco}, CALIC \cite{calic}, JPEG2000 \cite{skodras2001j2k}, WebP \cite{webp}, BPG \cite{bpg}, FLIF \cite{sneyers2016flif} and JPEG-XL \cite{alakuijala2019jpeg}, and nine recent learning-based lossless image compression methods including L3C \cite{Mentzer2019cvpr}, RC \cite{mentzer2020cvpr}, Bit-Swap \cite{kingma2019bitswap}, HiLLoC \cite{townsend2020hilloc}, IDF \cite{max2019nips}, IDF++ \cite{berg2020idfpp}, LBB \cite{ho2019compression}, iVPF \cite{zhang2021ivpf} and iFlow \cite{zhang2021iflow}.
For Bit-Swap, HiLLoC, IDF, IDF++ and LBB, their codes can hardly be applied on practical full resolution image compression tasks, and thus only be evaluated on ImageNet64 dataset. For RC, iVPF and iFlow, we report the compression performance published by their authors, because their codes are either difficult to be generalized to arbitrary datasets or unavailable.
We set $\lambda=0$ in \eqref{eq:loss_func} leading to the best lossless image compression performance.

As reported in Table\;\ref{tb:results_ll}, the proposed DLPR coding system achieves the best lossless compression performance on DIV2K validation dataset, which shares the same domain with the training dataset. The DLPR coding system also achieves the best compression performance on CLIC.p, CLIC.m, Kodak and Histo24 datasets, and achieves the second best compression performance on ImageNet64 validation dataset.
Though iFlow outperforms ours on ImageNet64 validation dataset, it is trained on ImageNet64 training dataset sharing the same domain while ours is trained on DIV2K.
The above results demonstrate that the DLPR coding system achieves the state-of-the-art lossless image compression performance and can be effectively generalized to images of various domains and modalities.

\subsection{Near-lossless Results of DLPR coding}
\begin{figure*}[!t]
\centering
\subfloat{
\includegraphics[width=.42\linewidth]{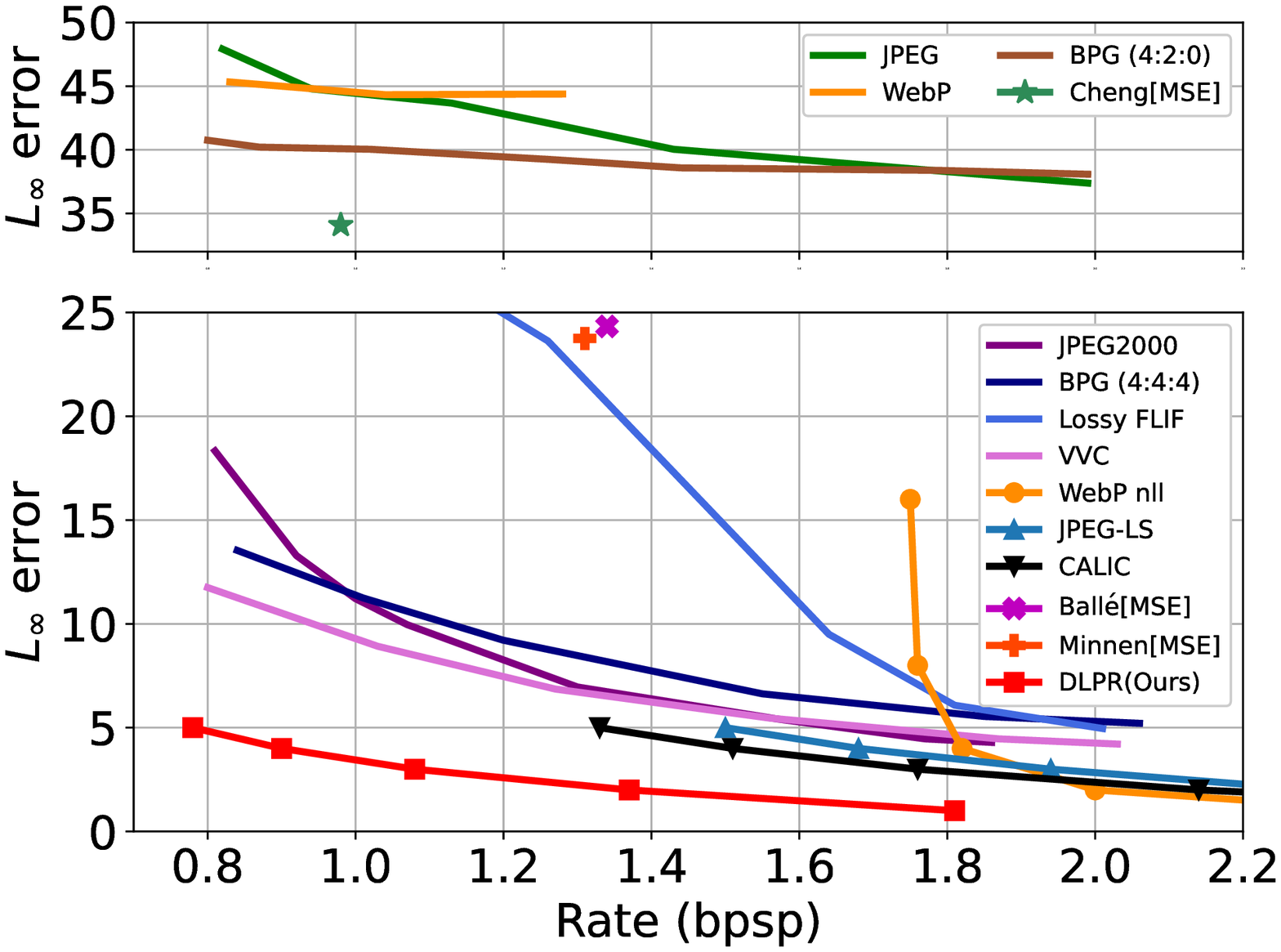}}
\subfloat{
\includegraphics[width=.42\linewidth]{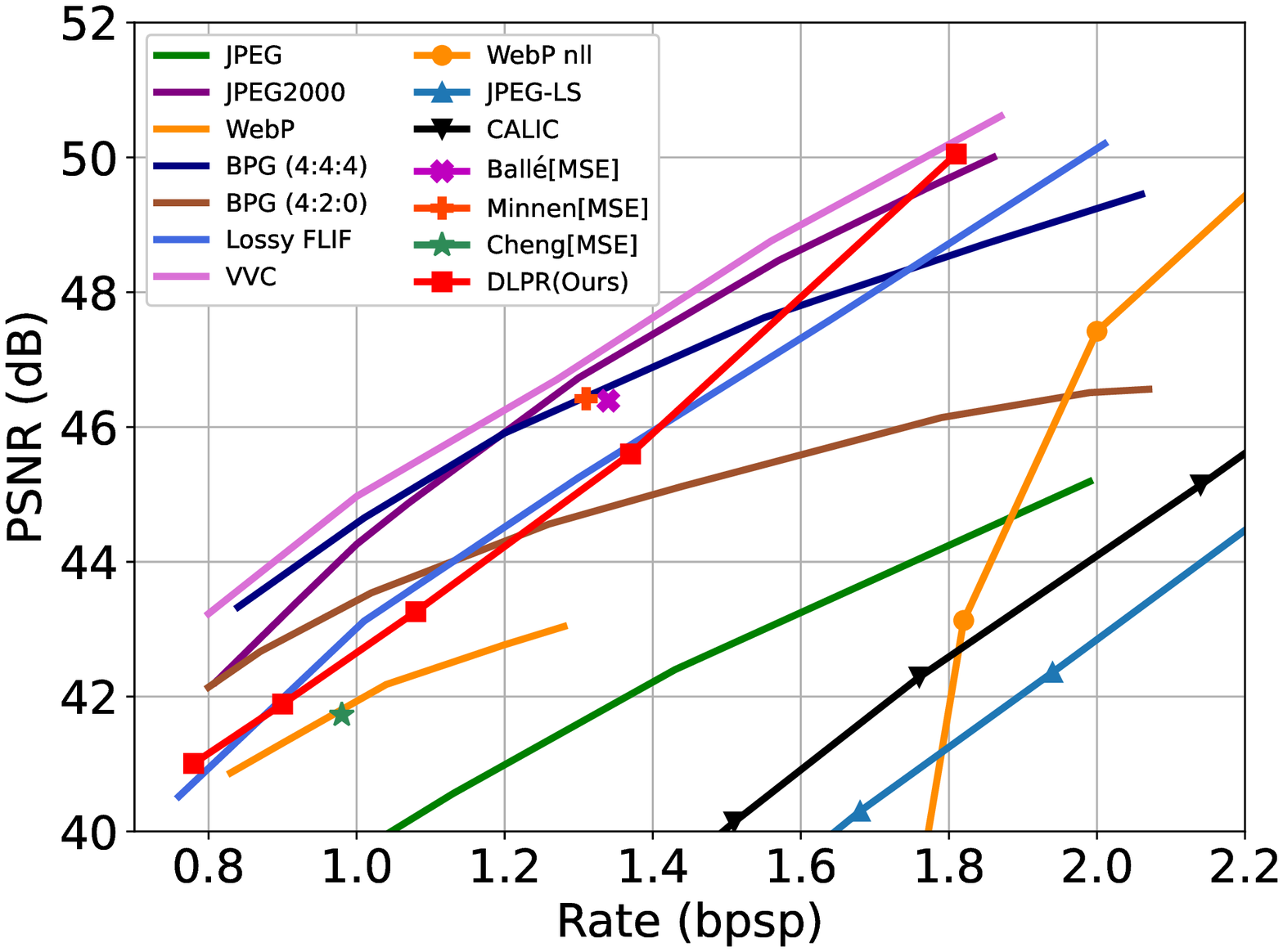}}
\caption{Rate-distortion performance of our DLPR coding system compared with other near-lossless image codecs and lossy image codecs on Kodak dataset.}
\label{fig:rd_kodak}
\end{figure*}

\begin{figure}[!t]
\begin{center}
\subfloat[Raw/Lossless]{
\label{fig:k07_t0}
\includegraphics[width=0.48\linewidth]{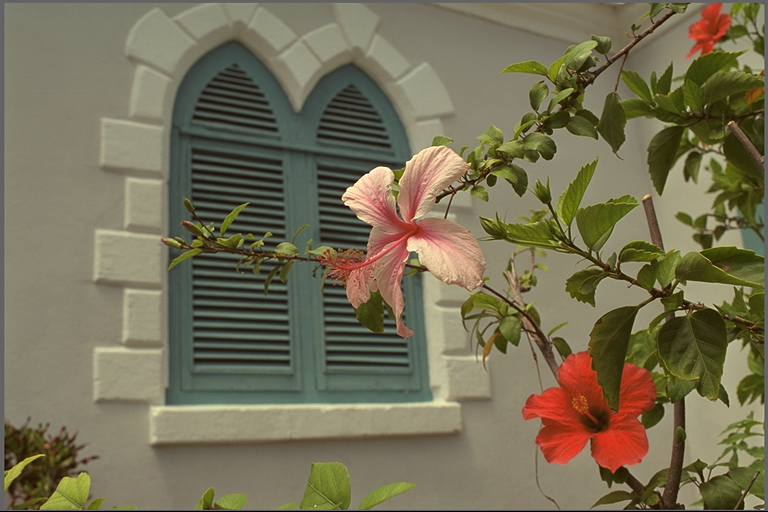}}
\subfloat[$\tau=1$]{
\label{fig:k07_t1}
\includegraphics[width=0.48\linewidth]{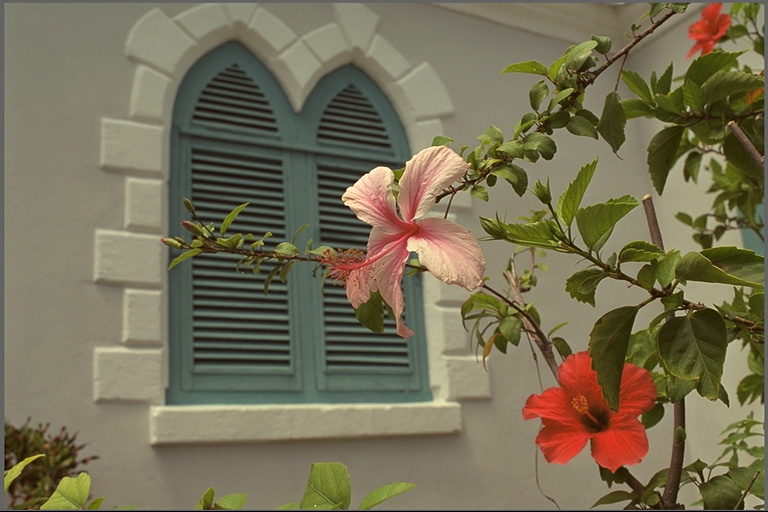}}\\
\subfloat[$\tau=2$]{
\label{fig:k07_t2}
\includegraphics[width=0.48\linewidth]{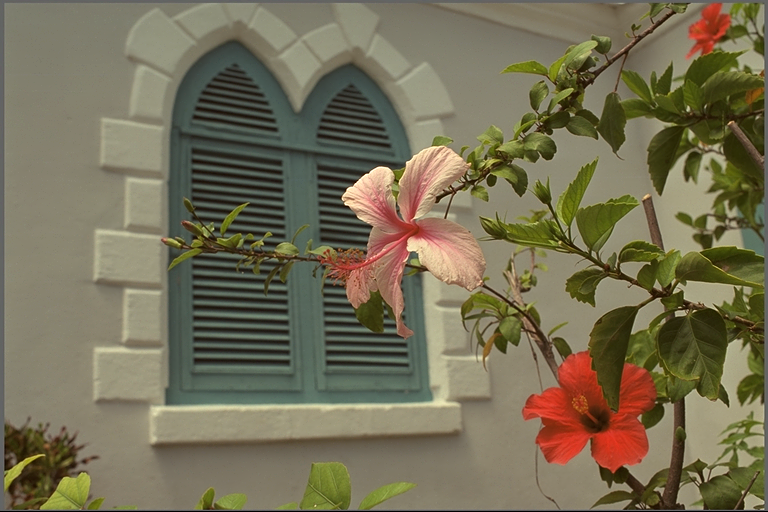}}
\subfloat[$\tau=4$]{
\label{fig:k07_t4}
\includegraphics[width=0.48\linewidth]{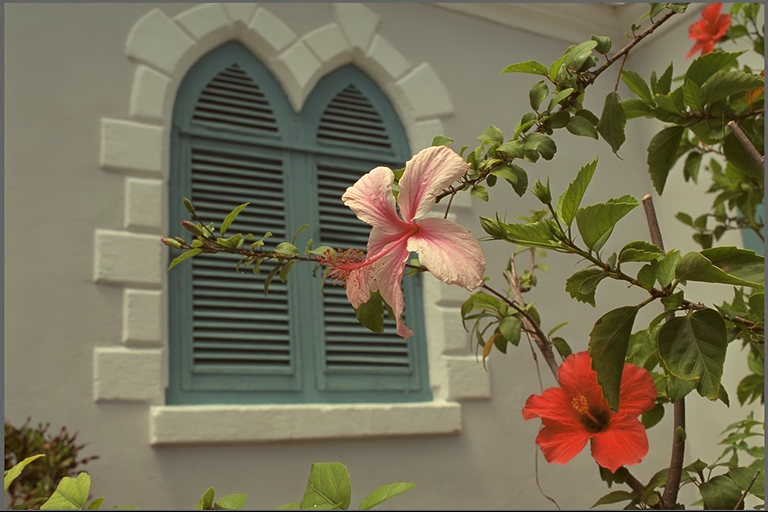}}
\end{center}
\caption{Near-lossless reconstructions of our DLPR coding system on Kodak dataset.}
\label{fig:subjective1}
\end{figure}

\begin{figure}[!t]
\begin{center}
\subfloat[Raw/Lossless]{
\label{fig:h1_t0}
\includegraphics[width=0.48\linewidth]{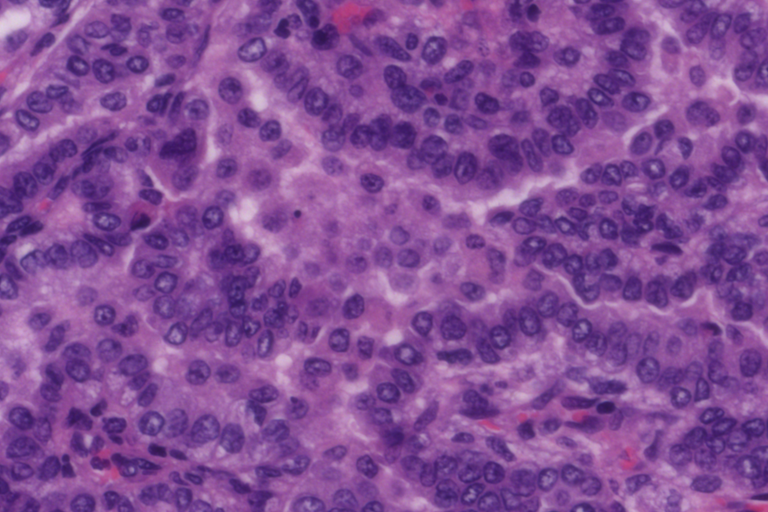}}
\subfloat[$\tau=1$]{
\label{fig:h1_t1}
\includegraphics[width=0.48\linewidth]{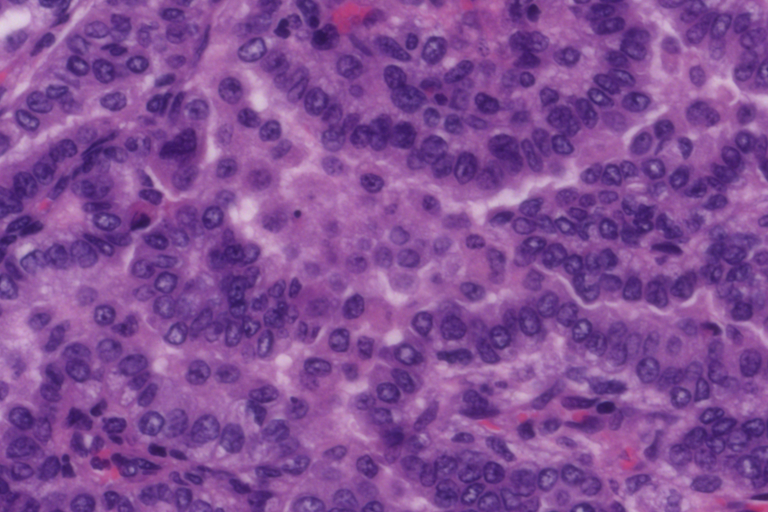}}\\
\subfloat[$\tau=2$]{
\label{fig:h1_t2}
\includegraphics[width=0.48\linewidth]{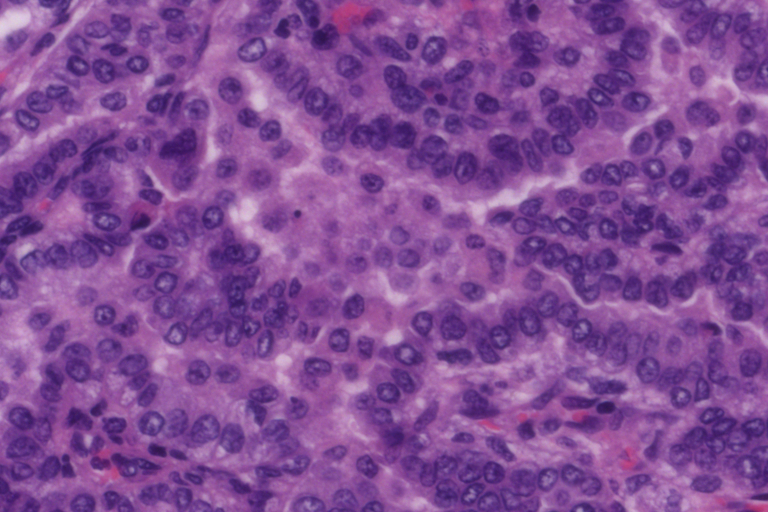}}
\subfloat[$\tau=4$]{
\label{fig:h1_t4}
\includegraphics[width=0.48\linewidth]{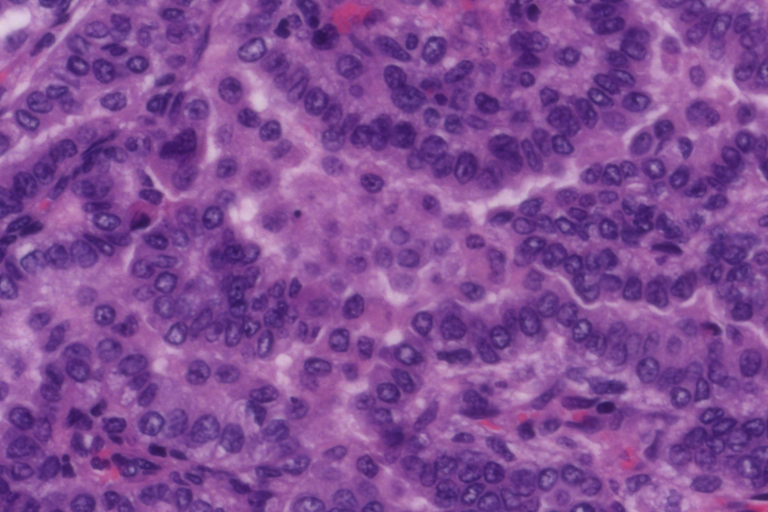}}
\end{center}
\caption{Near-lossless reconstructions of our DLPR coding system on Histo24 dataset.}
\label{fig:subjective2}
\end{figure}

We next evaluate the near-lossless image compression performance of the proposed DLPR coding system. We set $\lambda=0.03$ leading to the robust near-lossless image compression results with variable $\tau$'s. We compare with near-lossless WebP (WebP nll) \cite{webp}, near-lossless JPEG-LS \cite{weinberger2000loco} and near-lossless CALIC \cite{Wu2000nll}, as reported in Table\;\ref{tb:results_nll}. Near-lossless WebP adjusts pixel values to $\ell_\infty$ error bound $\tau$ and compresses the pre-processed images losslessly. Near-lossless JPEG-LS and CALIC adopt predictive coding schemes, and encode the residuals quantized by \eqref{eq:r_quantization}.
These three codecs handcraft the pre-processor, predictors and probability estimators, which are not efficient enough for variable $\tau$'s.
More efficiently, our DLPR coding system is based on jointly trained LIC, RC and SQRC. We employ \eqref{eq:r_quantization} to realize variable error bound $\tau$'s and the probability distributions of the quantized residuals are derived from the learned SQRC. Therefore, our DLPR coding system outperforms near-lossless WebP, JPEG-LS and CALIC by a wide margin.

Besides existing near-lossless image codecs, we also compare our near-lossless DLPR coding system with six traditional lossy image codecs, \emph{i.e.}, JPEG \cite{wallace1992jpeg}, JPEG2000 \cite{skodras2001j2k}, WebP \cite{webp}, BPG \cite{bpg}, Lossy FLIF \cite{sneyers2016flif} and VVC \cite{vvc}, and three representative learned lossy image compression methods, \emph{i.e.}, Ball\'{e}[MSE] \cite{Balle2018variational}, Minnen[MSE] \cite{minnen2018nips} and Cheng[MSE] \cite{cheng2020cvpr}, as shown in Fig.\;\ref{fig:rd_kodak}.
Because recent learned lossy image compression methods are all trained at relatively low bit rates ($\le$ 2 bpp $\approx$ 0.67 bpsp on Kodak), we re-implement Ball\'{e}[MSE], Minnen[MSE] and Cheng[MSE] at high bit-rates ($\ge$ 0.8 bpsp on Kodak).
Though Cheng[MSE] outperforms Ball\'{e}[MSE] and Minnen[MSE] at low bit rates, it performs worse than Ball\'{e}[MSE] and Minnen[MSE] at high bit rates because the sophisticated analysis and synthesis transforms hinder it from reaching very high-quality reconstructions. In terms of the rate-distortion performance measured by $\ell_\infty$ error, our DLPR coding consistently yields the best results among all codecs. Besides $\ell_\infty$ error, we also compare the rate-distortion performance of all codecs measured by PSNR. Our DLPR coding can achieve competitive performance at bit rates higher than 0.8 bpsp, even though PSNRs of near-lossless reconstructions are not our optimization objective.

In Fig.\;\ref{fig:subjective1} and \ref{fig:subjective2}, we display the near-lossless reconstructed images resulting from our DLPR coding at variable $\tau$'s on Kodak and Histo24 datasets. When $\tau$ is not large, human eyes can hardly differentiate between the raw images and the near-lossless reconstructions.

\subsection{Runtime of DLPR coding}
We evaluate the runtime of our DLPR coding on images of three different sizes. We compare with four representative traditional lossless image codecs including JPEG-LS \cite{weinberger2000loco}, BPG \cite{bpg}, FLIF \cite{sneyers2016flif} and JPEG-XL \cite{alakuijala2019jpeg}. We also compare with the practical learned lossless image compression method L3C \cite{Mentzer2019cvpr} and the learned lossy image compression method Minnen[MSE] \cite{minnen2018nips} with a serial autoregressive model. Note that the reported runtime includes both inference time and entropy coding time.

As reported in Table\;\ref{tb:runtime}, our lossless DLPR coding is almost as fast as FLIF with respect to encoding speed, much faster than BPG, JPEG-XL, L3C and lossy Minnen[MSE]. Although our lossless DLPR coding is slower than traditional codecs with respect to decoding speed, it is still practical for 2K resolution images and much faster than the learned L3C and lossy Minnen[MSE]. When $\tau>0$, the near-lossless DLPR coding can be even faster because the entropy coding time is reduced by the adaptive residual interval scheme.
Based on the above results, we demonstrates the practicability of the DLPR coding system and its great potential to be employed in real lossless and near-lossless image compression tasks.

\begin{table}[!t]
\caption{Runtime (sec.) of DLPR coding on images of different sizes (encoding/decoding time). For lossless compression ($\tau=0$), we use $\lambda=0$. For near-lossless compression ($\tau>0$), we use $\lambda=0.03$. \textit{*OOM denotes out-of-memory.}}
\label{tb:runtime}
\centering
\small
\begin{tabular}{lccc}
\toprule[1pt]
    Codec  & 768$\times$512  &  996$\times$756  & 2040$\times$1356  \tabularnewline
\midrule
     JPEG-LS \cite{weinberger2000loco}    &   0.12/0.12      &  0.23/0.22     & 0.83/0.80      \tabularnewline
     BPG   \cite{bpg}      &   2.38/0.13      &  4.46/0.27     & 16.52/0.98      \tabularnewline
     FLIF  \cite{sneyers2016flif}      &   0.90/0.16      &  1.84/0.35     & 7.50/1.35      \tabularnewline
     JPEG-XL \cite{alakuijala2019jpeg}    &   0.73/0.08      &  12.48/0.14    & 40.96/0.42      \tabularnewline
     L3C  \cite{Mentzer2019cvpr}       &   8.17/7.89      &  15.25/14.55   & OOM*      \tabularnewline
     Minnen[MSE] \cite{minnen2018nips} &   2.55/5.18      &  5.13/10.36    & 18.71/37.97      \tabularnewline
\midrule
     DLPR lossless    &   1.26/1.80   &  2.28/3.24     & 8.20/11.91  \tabularnewline
     DLPR ($\tau=1$)    &   0.79/1.24   &  0.98/1.86     & 2.09/5.56  \tabularnewline
     DLPR ($\tau=2$)    &   0.75/1.20   &  0.92/1.78     & 1.87/5.24  \tabularnewline
     DLPR ($\tau=4$)    &   0.73/1.18   &  0.89/1.74     & 1.68/5.03  \tabularnewline
\bottomrule[1pt]
\end{tabular}
\end{table}

\subsection{Ablation Study}
\begin{table*}[!t]
\caption{The relationships between different network architectures of LIC and lossless image compression performance (bpsp). \textit{Conv.} denotes $5\times 5$ convolutional layers used in \cite{Balle2018variational}. \textit{A/S Blks.} denotes the proposed analysis and synthesis blocks. \textit{Attn.} denotes the attention block used in \cite{cheng2020cvpr}. \textit{Swin Attn.} denotes the proposed Swin attention blocks.}
\label{tb:lossy_blks}
\centering
\small
\begin{tabular}{cc|cc|ccc}
\toprule[1pt]
    Conv. \cite{Balle2018variational} & A/S Blks.  &   Attn. \cite{cheng2020cvpr} & Swin-Attn. & ImageNet64 & DIV2K & CLIC.m \tabularnewline
\midrule
     $\checkmark$ &  $\times$       &  $\times$     & $\times$     &  3.75 (+0.06) & 2.60 (+0.05) & 2.20 (+0.04) \tabularnewline
     $\times$     &  $\checkmark$   &  $\times$     & $\times$     &  3.71 (+0.02) & 2.57 (+0.02) & 2.18 (+0.02) \tabularnewline
     $\times$     &  $\checkmark$   &  $\checkmark$ & $\times$     &  3.71 (+0.02) & 2.56 (+0.01) & 2.17 (+0.01) \tabularnewline
\midrule
     $\times$     &  $\checkmark$   &  $\times$     & $\checkmark$ &  3.69 & 2.55 & 2.16 \tabularnewline
\bottomrule[1pt]
\end{tabular}

\end{table*}

\begin{table}[!t]
\caption{The relationships between different network architectures of RC and lossless image compression performance (bpsp).}
\label{tb:u_and_cr}
\centering
\small
\begin{tabular}{cc|ccc}
\toprule[1pt]
    $\u$  & $C_\r$ & ImageNet64 & DIV2K & CLIC.m \tabularnewline
\midrule
     $\checkmark$ &  $\times$            &  4.03 (+0.34) & 2.91 (+0.36) & 2.47 (+0.31) \tabularnewline
     $\times$     &  $\checkmark$        &  3.78 (+0.09) & 2.64 (+0.09) & 2.24 (+0.08) \tabularnewline
\midrule
     $\checkmark$ &  $\checkmark$        &  3.69 & 2.55 & 2.16 \tabularnewline
\bottomrule[1pt]
\end{tabular}

\end{table}

\noindent\textbf{Network architectures of LIC and RC.}
In Table\;\ref{tb:lossy_blks}, we study the relationships between different network architectures of LIC and lossless compression performance. Compared with the $5\times 5$ convolutional layers used in \cite{Balle2018variational} and the attention blocks used in \cite{cheng2020cvpr}, the proposed analysis/synthesis blocks and Swin attention blocks can effectively improve the lossless image compression performance. These results also demonstrate that LIC plays an important role in our DLPR coding system.

In Table\;\ref{tb:u_and_cr}, we study the relationships between different network architectures of RC and lossless image compression performance. Both the feature $\u$ and context $C_\r$ can effectively improve the lossless image compression performance, which demonstrates the effectiveness of the proposed network architecture of RC.
In Table.\;\ref{tb:llm_gaussian}, we further compare the lossless and near-lossless image compression performance of the logistic mixture model, Gaussian single model and Gaussian mixture model for RC and SQRC.
The logistic mixture model and Gaussian mixture model achieve almost identical performance, outperforming the Gaussian single model for complex real distributions of $\r$ and $\hr$.
We utilize the logistic mixture model because its cumulative distribution function (CDF) is a sigmoid function, making it easier to compute the probability of each discrete residual using \eqref{eq:r_eval}.
In contrast, the CDF of the Gaussian distribution is the more complex Gauss error function.

\begin{table}[htb]
\caption{Lossless and near-lossless image compression performance (bpsp) resulting from logistic mixture model (lmm.) with $K=5$, Gaussian single model (gsm.) and Gaussian mixture model (gmm.) with $K=5$.}
\label{tb:llm_gaussian}
\centering
\small
\begin{tabular}{l|cccc}
\toprule[1pt]
    model & ImageNet64 & DIV2K & CLIC.m \tabularnewline
\midrule
      lmm.(lossless)        &  3.69  & 2.55  & 2.16  \tabularnewline
      gsm.\;(lossless)      &  3.78 (+0.09) & 2.57 (+0.02) & 2.22 (+0.06)\tabularnewline
      gmm.(lossless)        &  3.70 (+0.01) & 2.55 (=) & 2.17 (+0.01)\tabularnewline
\midrule
      lmm.($\tau=1$)         &  2.59 & 1.69 & 1.50 \tabularnewline
      gsm.\;($\tau=1$)      &  2.61 (+0.02) & 1.72 (+0.03) & 1.52 (+0.02)\tabularnewline
      gmm.($\tau=1$)       &  2.59 (=) & 1.68 ($-$0.01) & 1.50 (=)\tabularnewline
\bottomrule[1pt]
\end{tabular}

\end{table}

\noindent\textbf{$\tau$-specific vs. scalable.}
As aforementioned in Sec.\;\ref{subsubsec:sqrc}, $\tau$-specific near-lossless image compression scheme leads to the challenging relaxation problem of residual quantization. In order to compare with our scalable near-lossless image compression scheme, we realize the $\tau$-specific models by relaxing residual quantization with straight-through (copying gradients from quantized $\hr$ to original $\r$). As shown in Fig.\;\ref{fig:tau_specific}, the resulting $\tau$-specific models perform worse than our scalable model in most cases, due to the gradient bias during training.

\begin{figure}[!t]
\centering
\includegraphics[width=.9\linewidth]{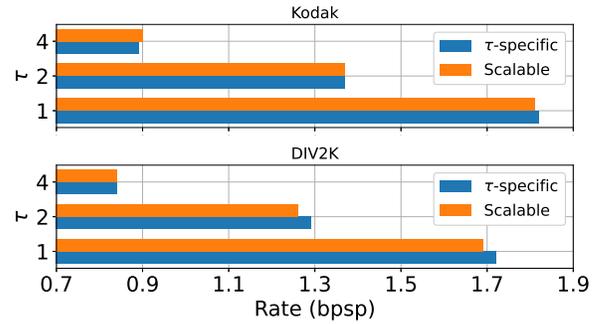}
\caption{Comparisons between $\tau$-specific scheme and the proposed scalable near-lossless compression scheme.}
\label{fig:tau_specific}
\end{figure}

In Fig.\;\ref{fig:tau_scalable}, we show lossless to near-lossless image compression performance $\tau\in\{0, 1, . . . , 5\}$ of each images compressed by our scalable model with $\lambda=0.03$ on Kodak dataset. $R_{\hy,\hz}$, on average, accounts for about 16\% of $R_{\hy,\hz}+R_\r$ at $\tau=0$. With the increase of $\tau$, the bit-rate $R_{\hr}^\tau$ of the quantized residual $\hr$ is significantly reduced.
Especially the $\tau=1$ near-lossless mode saves about 39\% bit rates compared with the $\tau=0$ lossless mode.

\begin{figure}[!t]
\centering
\includegraphics[width=.9\linewidth]{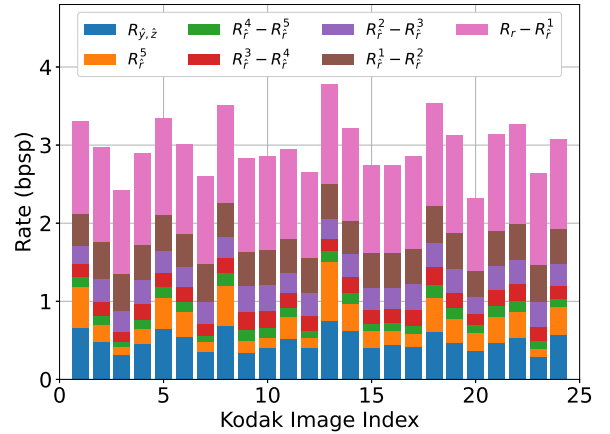}
\caption{Lossless to near-lossless image compression performance $\tau\in\{0, 1, . . . , 5\}$ of each image compressed by the DLPR coding system with $\lambda=0.03$ on Kodak dataset.}
\label{fig:tau_scalable}
\end{figure}

\noindent\textbf{SQRC for bias correction.}
In Fig.\;\ref{fig:bias_correction}, we demonstrate the efficacy of SQRC for bias correction. Because of the discrepancy between the oracle $\hp_\btheta(\hr|\u,C_\r)$ and the biased $\hp_\btheta(\hr|\u,C_\hr)$, encoding $\hr$ with $\hp_\btheta(\hr|\u,C_\hr)$ (without SQRC) degrades the compression performance. Instead, we encode $\hr$ with $\hp_\bvarphi(\hr|\u, C_{\hr},\tau)$ (with SQRC) resulting in lower bit rates.
With the increase of $\tau$, the PMF of $\r$ is quantized by larger bins and becomes coarser. Thus, the compression performance with SQRC approaches the oracle. However, the gap between the compression performance without SQRC and the oracle remains large, as the
discrepancy between $\hp_\btheta(\hr|\u,C_\r)$ and $\hp_\btheta(\hr|\u,C_\hr)$ is also magnified with the increasing $\tau$. Note that the compression performance of our DLPR coding without SQRC is still better than near-lossless WebP, JPEG-LS and CALIC.

\begin{figure}[!t]
\centering
\includegraphics[width=.9\linewidth]{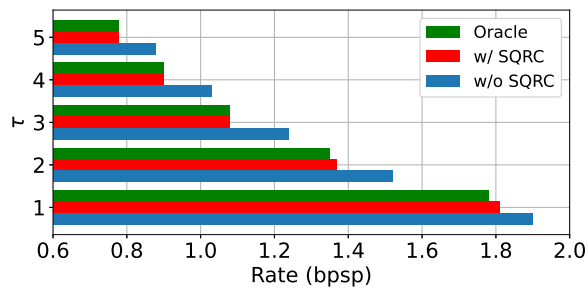}
\caption{Near-lossless image compression performance of scalable DLPR coding system with oracle $\hat{p}_\btheta(\hr|\u, C_\r)$, with SQRC and without SQRC.}
\label{fig:bias_correction}
\end{figure}

\noindent\textbf{Discussion on different $\lambda$'s.}
The $\lambda$'s in \eqref{eq:loss_func} adapts the lossless compression rate of the raw image $\x$ and the distortion of the lossy reconstruction $\tx$. In Table\;\ref{tb:lambda_ablation}, we evaluate the effects of different $\lambda$'s $\in\{0, 0.001, 0.03, 0.06\}$ on lossless image compression performance on Kodak dataset. With the increase of $\lambda$, both the rate $R_{\hy,\hz}$ and the PSNR of $\tx$ become higher. While the rate $R_\r$ of residual decreases at the same time, the decrease of $R_\r$ is smaller than the increase of $R_{\hy,\hz}$, leading to the degradation of lossless image compression performance. $\lambda=0$ leads to the best lossless compression performance. We visualize lossy reconstruction $\tx$'s, residual $\r$'s and feature $\u$'s resulting from different $\lambda$'s in Fig.\;\ref{fig:visualization_lambda}. Interestingly, DLPR coding with $\lambda=0$ learns to set $\tx=\mathbf{0}$ and $\r=\x$. The LIC becomes a special feature compressor extracting feature $\u$ for lossless image compression (proved effective in Table\;\ref{tb:u_and_cr}). This special case of DLPR coding with $\lambda=0$ is similar to PixelVAE \cite{gulrajani2016pixelvae}.

In Fig.\;\ref{fig:lambda_ablation}, we further study the effects of different $\lambda$'s on the near-lossless image compression performance.
Though $\lambda=0$ leads to the best lossless compression performance, it is unsuitable for near-lossless compression since the residual quantization \eqref{eq:r_quantization} is adopted on the $\r=\x$. For near-lossless compression, we set $\lambda=0.03$. Compared with $\lambda=0$ and $0.001$, the quantized residual $\hr$ of $\lambda=0.03$ results in much lower entropy and smaller context bias, since most elements of $\r$ and $\hr$ are zeros or close to zeros.
The reduction of $R_\hr^\tau$ of $\lambda=0.03$ compensates for larger $R_{\hy,\hz}$ with the increase of $\tau$.
Compared with $\lambda=0.06$, $\lambda=0.03$ enjoys similar $R_\hr^\tau$ but lower $R_{\hy,\hz}$.
Therefore, $\lambda=0.03$ achieves the most robust near-lossless image compression performance only slightly worse than $\lambda=0$ and $0.001$ at $\tau=1$.

\begin{table}[!t]
\caption{The effects of different $\lambda$'s on the lossless image compression performance on Kodak dataset.}
\label{tb:lambda_ablation}
\centering
\small
\begin{tabular}{lcC{3em}C{3em}c}
\toprule[1pt]
     $\lambda$ & Total Rate & $R_{\hy,\hz}$ & $R_{\r}$ & $\tx$ (PSNR) \tabularnewline
\midrule
     0        &  2.86     &  0.04      &  2.82   &  6.78  \tabularnewline
     0.001    &  2.94     &  0.13      &  2.81   &  29.80  \tabularnewline
     0.03     &  2.99     &  0.49      &  2.50   &  38.77  \tabularnewline
     0.06     &  3.02     &  0.59      &  2.43   &  39.74  \tabularnewline
\bottomrule[1pt]
\end{tabular}

\end{table}

\begin{figure}[!t]
\centering
\includegraphics[width=.9\linewidth]{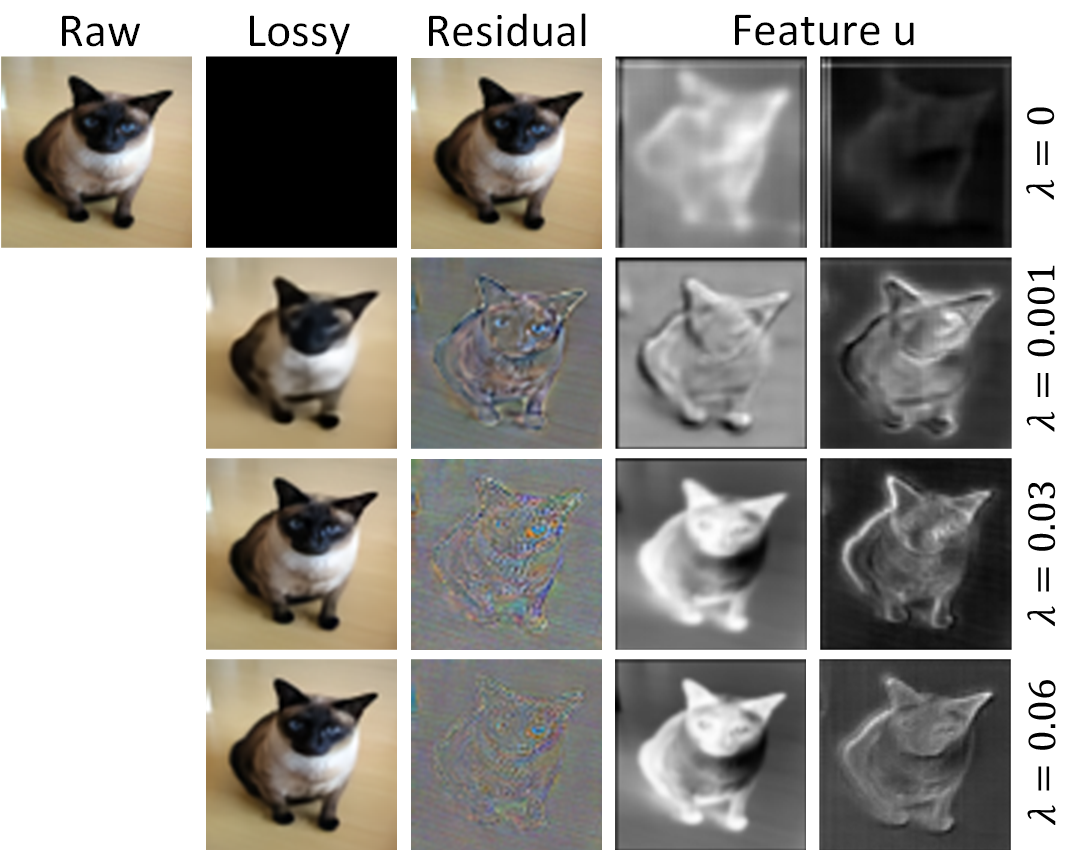}
\caption{Visual examples of lossy reconstructions, residuals and feature $\u$'s with different $\lambda$'s. When $\lambda=0$, the LIC becomes a special feature compressor extracting feature $\u$ for lossless image compression. The DLPR coding with $\lambda=0$ is similar to PixelVAE \cite{gulrajani2016pixelvae}.}
\label{fig:visualization_lambda}
\end{figure}

\begin{figure}[!t]
\centering
\includegraphics[width=.9\linewidth]{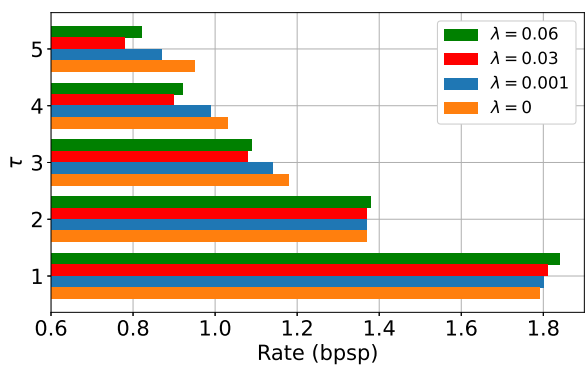}
\caption{The effects of different $\lambda$'s on the near-lossless image compression performance on Kodak dataset.}
\label{fig:lambda_ablation}
\end{figure}

\begin{figure}[htb]
\centering
\includegraphics[width=.9\linewidth]{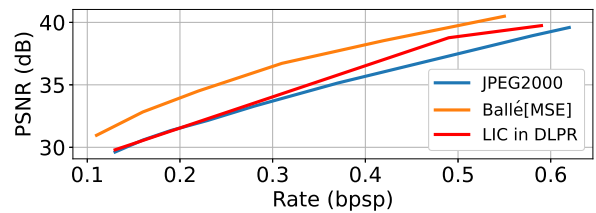}
\caption{Rate-distortion performance of LIC in DLPR coding at different $\lambda$'s on Kodak dataset.}
\label{fig:lic_comparison}
\end{figure}

\begin{table*}[!t]
\caption{Lossless image compression performance (bpsp) of different context models on ImageNet64, DIV2K, CLIC.p, CLIC.m, Kodak and Histo24 datasets. The context model $M_7^5$ is set as the anchor.}
\label{tb:context_design}
\centering
\small
\begin{tabular}{lC{5.5em}C{5.5em}C{5.5em}C{5.5em}C{5.5em}C{5.5em}}
\toprule[1pt]
  Context ($7\times 7$)    &  ImageNet64 & DIV2K & CLIC.p & CLIC.m & Kodak & Histo24 \tabularnewline
\midrule
    $M_7^5$ (Fig.\;\ref{fig:coding_seq} \& \ref{fig:coding_acc1})  & 3.69  & 2.55 & 2.37 & 2.15 & 2.86 & 2.96    \tabularnewline
    $M_7^4$ (Fig.\;\ref{fig:coding_acc2})                          & 3.70 (+0.01)  & 2.56 (+0.01) & 2.39 (+0.02) & 2.16 (+0.01) & 2.87 (+0.01) & 2.97 (+0.01)   \tabularnewline
\midrule
    $M_7^3$ (Fig.\;\ref{fig:coding_acc3}, Selected)                   & 3.69 (=)  & 2.55 (=) & 2.38 (+0.01) & 2.16 (+0.01) & 2.86 (=) & 2.96 (=) \tabularnewline
\midrule
    $M_7^2$ (Fig.\;\ref{fig:coding_acc4})                       & 3.71 (+0.02) & 2.59 (+0.04) & 2.40 (+0.03) & 2.19 (+0.04) & 2.95 (+0.09) & 3.01 (+0.05)   \tabularnewline
    $M_7^1$ (Fig.\;\ref{fig:coding_acc5})                     & 3.88 (+0.19) & 2.74 (+0.19) & 2.56 (+0.19) & 2.34 (+0.19) & 2.95 (+0.09) & 3.24 (+0.28)   \tabularnewline
\midrule
    checkerboard                  & 3.86 (+0.17) & 2.73 (+0.18) & 2.54 (+0.17) & 2.33 (+0.18) & 2.95 (+0.09) & 3.19 (+0.23)  \tabularnewline
    channel-only                  & 4.03 (+0.34) & 2.91 (+0.36) & 2.70 (+0.33) & 2.47 (+0.32) & 3.09 (+0.23) & 3.48 (+0.52)  \tabularnewline
    w/o context                   & 4.47 (+0.78) & 3.37 (+0.82) & 3.19 (+0.82) & 3.14 (+0.99) & 3.61 (+0.75) & 3.45 (+0.49) \tabularnewline
\bottomrule[1pt]
\end{tabular}

\end{table*}

\begin{table}[!t]
\caption{Runtime (sec.) of lossless DLPR coding with different context models and adaptive residual interval (AdaRI.) on Kodak dataset.}
\label{tb:runtime_ablation}
\centering
\small
\begin{tabular}{lccc}
\toprule[1pt]
    Context (7$\times$7)  &  AdaRI.       & Enc./Dec. Time \tabularnewline
\midrule
   $M_7^5$ (Fig.\;\ref{fig:coding_seq}, Serial)      & $\times$          & 12.46/23.24   \tabularnewline
   $M_7^5$ (Fig.\;\ref{fig:coding_seq}, Serial)      & $\checkmark$      & 11.93 (-0.53)/22.74 (-0.50)  \tabularnewline
   $M_7^5$ (Fig.\;\ref{fig:coding_acc1})      & $\checkmark$  & 1.47 (-10.99)/2.30 (-20.94)   \tabularnewline
\midrule
   $M_7^3$ (Fig.\;\ref{fig:coding_acc3}, Selected)     & $\checkmark$  & 1.24 (-11.22)/1.75 (-21.49)   \tabularnewline
\bottomrule[1pt]
\end{tabular}

\end{table}

\noindent\textbf{Rate-distortion performance of LIC.} We conduct an ablation study to show the rate-distortion performance of our LIC in DLPR coding framework in Fig.\;\ref{fig:lic_comparison}.
Besides lossy reconstruction $\tx$, our LIC also generates feature $\u$ of residual $\r$ and is jointly trained with RC. Therefore, the rate $R_{\hy,\hz}$ of our LIC not only carries the information from lossy reconstruction $\tilde{\x}$ but also carries part of the information from residual $\r$.
As a result, the rate-distortion performance of our LIC in DLPR coding framework is between Ball\'e[MSE] \cite{Balle2018variational} and JPEG2000 on Kodak dataset.

\noindent\textbf{Context design and adaptive residual interval.}
In Table\;\ref{tb:context_design}, we study the lossless image compression performance resulting from the designed $7\times 7$ context models in Fig.\;\ref{fig:coding_acceleration}. The context model $M_7^5$ is set as the anchor. Based on the experimental results, the selected context model $M_7^3$ removing two upper-right context pixels from $M_7^5$ achieves almost the same lossless image compression performance as $M_7^5$, even slightly better than $M_7^4$. At the same time, $M_7^3$ reduces about 40\% parallel coding steps from $5P-4$ to $3P-2$, compared with $M_7^5$. Though the context model $M_7^2$ (zigzag) and $M_7^1$ can be faster, the corresponding lossless compression performance drops significantly.
Besides, we also show the compression performance resulting from checkerboard context model \cite{he2021cvpr}. Without the help of transform coding, the compression performance of checkerboard context model is only slightly better than context $M_7^1$ and worse than context $M_7^2$ in our DLPR coding framework. We further switch off spatial context models and show the compression performance resulting from only channel-wise autoregressive scheme in \eqref{eq:r_chain_rule} and \eqref{eq:channel_ar}. The channel-only setting cannot reduce spatial redundancies among residuals, resulting in worse compression performance.
We also evaluate the compression performance without both spatial and channel context models.
Compared with the channel-only setting, the w/o context setting ignoring channel redundancies results in the worst compression performance, except on Histo24.
This is because the color properties of stained histological images differ from natural images, and the channel autoregressive model trained on DIV2K does not generalize as well on Histo24 as on other datasets.
In Table\;\ref{tb:runtime_ablation}, we finally shows that the designed context models and the adaptive residual interval scheme can effectively reduce the runtime of DLPR coding in practice.

\section{Conclusion}
\label{sec:conclusion}
In this paper, we propose a unified DLPR coding framework for both lossless and near-lossless image compression. The DLPR coding framework consists of a lossy image compressor, a residual compressor and a scalable quantized residual compressor, which is formulated in terms of VAEs and is solved with end-to-end training. The DLPR coding framework supports scalable near-lossless image compression with variable $\ell_\infty$-constraint $\tau$'s in a single network, instead of multiple networks for different $\tau$'s.
We further propose a novel design of context coding and an adaptive residual interval scheme to significantly accelerate the coding process.
Extensive experiments demonstrate that the DLPR coding system achieves not only the state-of-the-art compression performance, but also competitive coding speed for practical full resolution image compression tasks.

% or


%\appendices
%\input{appendices}

%\section{}

%\ifCLASSOPTIONcompsoc
%  % The Computer Society usually uses the plural form
%\else
%  % regular IEEE prefers the singular form
%\fi

\ifCLASSOPTIONcaptionsoff
  \newpage
\fi



\bibliographystyle{IEEEtran}
\bibliography{nll_jrnl}
%\begin{thebibliography}{1}
%\end{thebibliography}


%\end{IEEEbiography}

\begin{IEEEbiography}[{\includegraphics[width=1in,height=1.25in,clip,keepaspectratio]{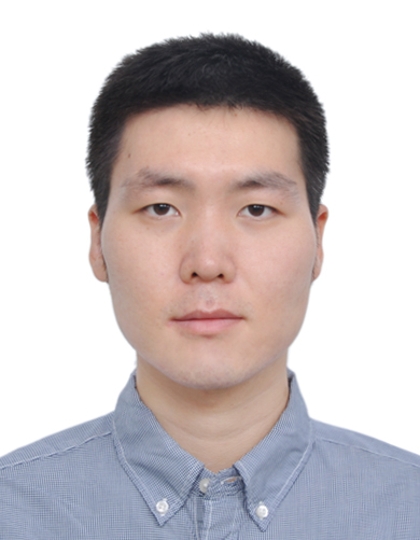}}]{Yuanchao Bai} (Member, IEEE)
received the B.S. degree in software engineering from Dalian University of Technology, Liaoning, China, in 2013. He received the Ph.D. degree in computer science from Peking University, Beijing, China, in 2020. He was a postdoctoral fellow in Peng Cheng Laboratory, Shenzhen, China, from 2020 to 2022. He is currently an assistant professor  with the School of Computer Science and Technology, Harbin Institute of Technology, Harbin, China.
His research interests include image/video compression and processing, deep unsupervised learning, and graph signal processing.
\end{IEEEbiography}

\begin{IEEEbiography}[{\includegraphics[width=1in,height=1.25in,clip,keepaspectratio]{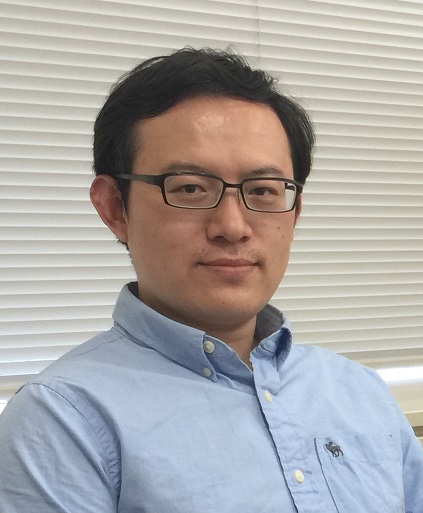}}]{Xianming Liu} (Member, IEEE)
 is a Professor with the School of Computer Science and Technology, Harbin Institute of Technology (HIT), Harbin, China. He received the B.S., M.S., and Ph.D degrees in computer science from HIT, in 2006, 2008 and 2012, respectively. In 2011, he spent half a year at the Department of Electrical and Computer Engineering, McMaster University, Canada, as a visiting student, where he then worked as a post-doctoral fellow from December 2012 to December 2013. He worked as a project researcher at National Institute of Informatics (NII), Tokyo, Japan, from 2014 to 2017. He has published over 60 international conference and journal publications, including top IEEE journals, such as T-IP, T-CSVT, T-IFS, T-MM, T-GRS; and top conferences, such as CVPR, IJCAI and DCC. He is the receipt of IEEE ICME 2016 Best Student Paper Award.
\end{IEEEbiography}

\begin{IEEEbiography}[{\includegraphics[width=1in,height=1.25in,clip,keepaspectratio]{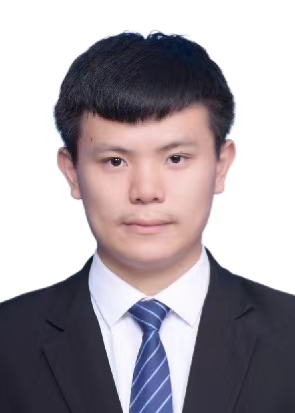}}]{Kai Wang}
received the B.S. degree in software engineering from Harbin Engineering University, Harbin, China, in 2020 and received the M.S. degree of electronic information in software engineering from Harbin Institute of Technology, Harbin, China, in 2022. He is currently pursuing the docter degree in electronic information in Harbin Institute of Technology, Harbin, China.
His research interests include image/video compression and deep learning.
\end{IEEEbiography}

\begin{IEEEbiography}[{\includegraphics[width=1in,height=1.25in,clip,keepaspectratio]{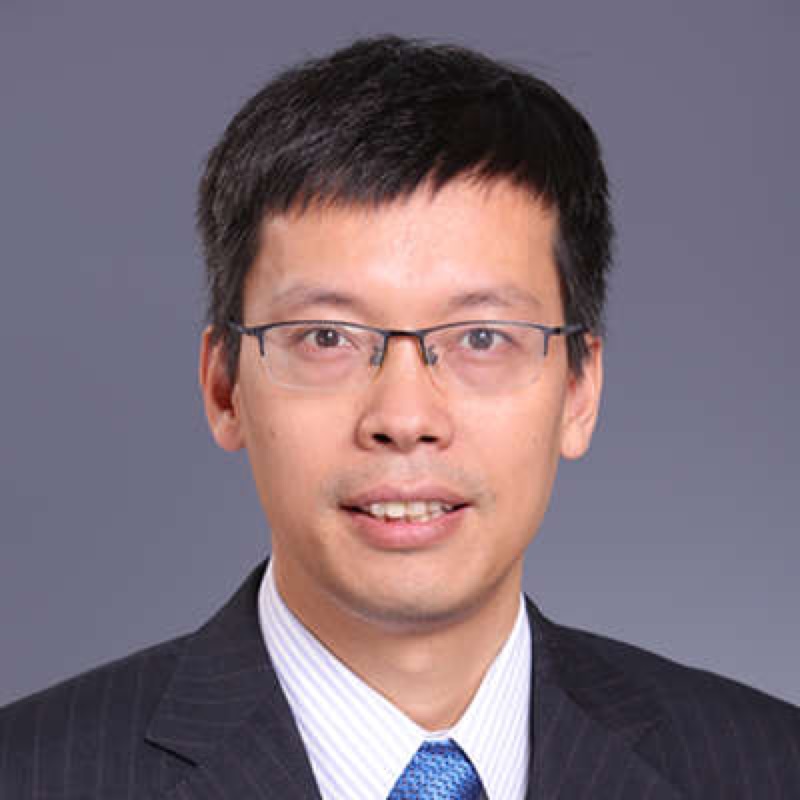}}]{Xiangyang Ji} (Member, IEEE) received the B.S. degree in materials science and the M.S. degree in computer science from the Harbin Institute of Technology, Harbin, China, in 1999 and 2001, respectively, and the Ph.D. degree in computer science from the Institute of Computing Technology, Chinese Academy of Sciences, Beijing, China. He joined Tsinghua University, Beijing, in 2008, where he is currently a Professor with the Department of Automation, School of Information Science and Technology. He has authored over 100 referred conference and journal papers. His current research interests include signal processing, image/video compressing, and intelligent imaging.
\end{IEEEbiography}

\begin{IEEEbiography}[{\includegraphics[width=1in,height=1.25in,clip,keepaspectratio]{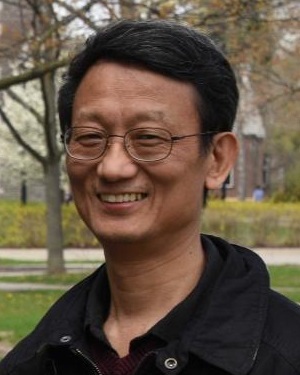}}]{Xiaolin Wu} (Fellow, IEEE) received the B.Sc. degree in computer science from Wuhan University, China, in 1982, and the Ph.D. degree in computer science from the University of Calgary, Canada, in 1988. He started his academic career in 1988. He was a Faculty Member with Western University, Canada, and New York Polytechnic University (NYU-Poly), USA. He is currently with McMaster University, Canada, where he is a Distinguished Engineering Professor and holds an NSERC Senior Industrial Research Chair. His research interests include image processing, data compression, digital multimedia, low-level vision, and network-aware visual communication. He has authored or coauthored more than 300 research articles and holds four patents in these fields. He served on technical committees of many IEEE international conferences/workshops on image processing, multimedia, data compression, and information theory. He was a past Associated Editor of IEEE TRANSACTIONS ON MULTIMEDIA. He is also an Associated Editor of IEEE TRANSACTIONS ON IMAGE PROCESSING.
\end{IEEEbiography}

\begin{IEEEbiography}[{\includegraphics[width=1in,height=1.25in,clip,keepaspectratio]{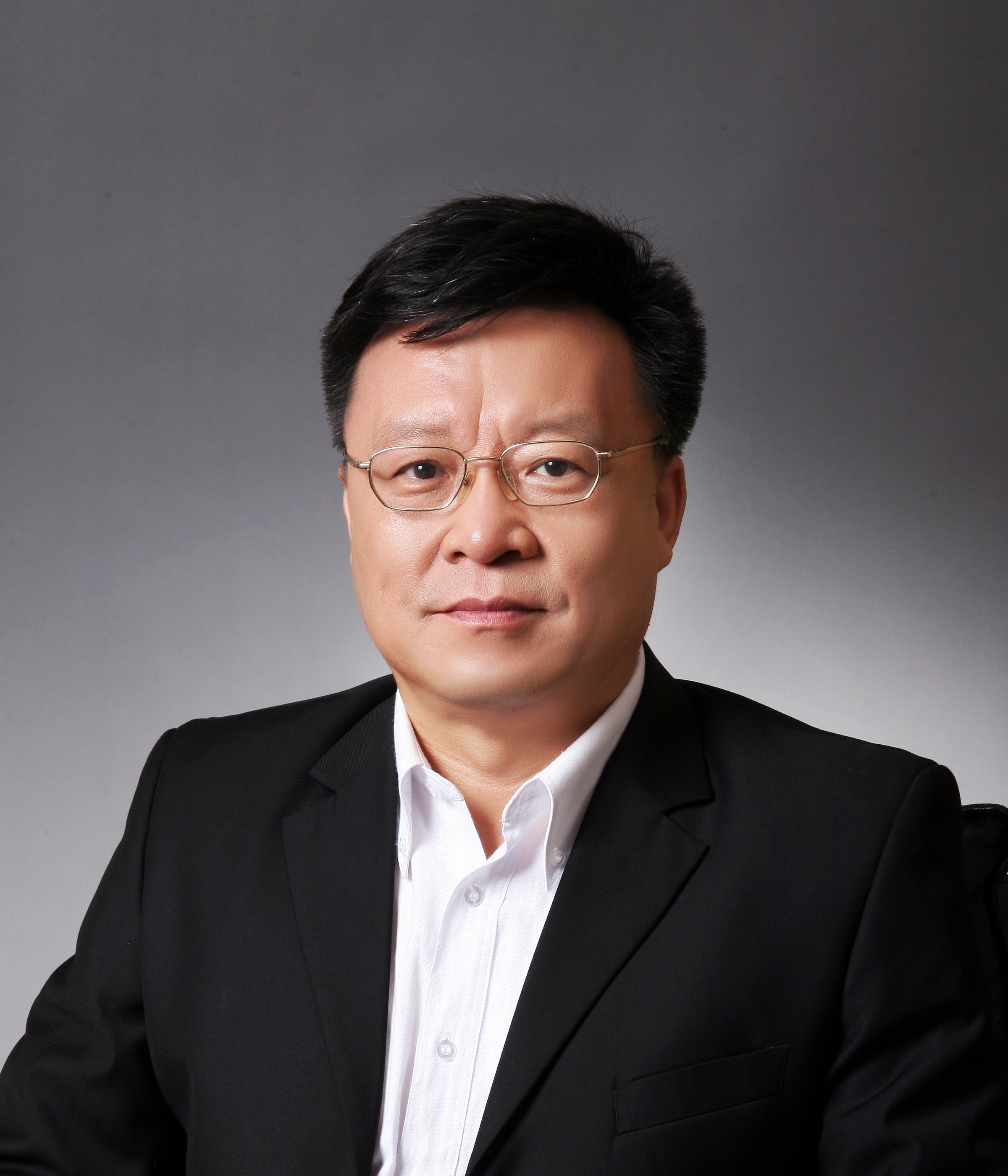}}]{Wen Gao} (Fellow, IEEE) received the Ph.D. degree in electronics engineering from The University of Tokyo, Japan, in 1991. He is currently a Boya Chair Professor in computer science at Peking University. He is the Director of Peng Cheng Laboratory, Shenzhen. Before joining Peking University, he was a Professor with Harbin Institute of Technology from 1991 to 1995. From 1996 to 2006, he was a Professor at the Institute of Computing Technology, Chinese Academy of Sciences. He has authored or coauthored five books and over 1000 technical articles in refereed journals and conference proceedings in the areas of image processing, video coding and communication, computer vision, multimedia retrieval, multimodal interface, and bioinformatics. He served on the editorial boards for several journals, such as ACM CSUR, IEEE TRANSACTIONS ON IMAGE PROCESSING (TIP), IEEE TRANSACTIONS ON CIRCUITS AND SYSTEMS FOR VIDEO TECHNOLOGY (TCSVT), and IEEE TRANSACTIONS ON MULTIMEDIA (TMM). He served on the advisory and technical committees for professional organizations. He was the Vice President of the National Natural Science Foundation (NSFC) of China from 2013 to 2018 and the President of China Computer Federation (CCF) from 2016 to 2020. He is the Deputy Director of China National Standardization Technical Committees. He is an Academician of the Chinese Academy of Engineering and a fellow of ACM. He chaired a number of international conferences, such as IEEE ICME 2007, ACM Multimedia 2009, and IEEE ISCAS 2013.
\end{IEEEbiography}


%\newpage



%\vfill


\end{document}